\documentclass[12pt]{article}
\usepackage{amsmath}
\usepackage{graphicx}
\usepackage{natbib}
\usepackage{url} 
\usepackage{amsmath, amsthm, amssymb, natbib, latexsym,algorithm,multirow,enumerate,url,booktabs,xcolor,bbm,soul, sectsty, epsfig, cancel, makecell}

\newcommand{\blind}{0}
\usepackage[title]{appendix}
\usepackage{xcolor}
\usepackage{soul}
\usepackage{bm}

\usepackage{caption}
\def\e{{\rm e}}                        
\def\tr{\text{\rm tr}}
\def\vec{\text{\rm vec}}
\def\d{\text{\rm d}}

\def\Var{\text{\rm Var}}
\def\logit{\text{\rm logit}}

\def\diag{\text{\rm diag}}
\def\dg{\text{\rm dg}}
\def\vech{\text{\rm vech}}

\newcommand\Poisson{\text{Poisson}}

\def\mG{{\mathcal{G}}}

\def\N{{\text{N}}}
\def\IG{{\text{IG}}}
\def\E{{\text{E}}}
\def\SN{{\text{SN}}}

\def\CSN{{\text{CSN}}}
\def\bfone{{\bf 1}}

\def\K{\mathsf{K}}
\def\Ed{\mathsf{E}_d}
\def\El{\mathsf{E}_\ell}
\def\Eu{\mathsf{E}_u}

\newcommand\mL{{\mathcal{L}}}

\newtheorem{theorem}{Theorem}

\newtheorem{lemma}{Lemma}

\graphicspath{{figures/}}

\begin{document}

\def\spacingset#1{\renewcommand{\baselinestretch}%
{#1}\small\normalsize} \spacingset{1}


\if0\blind
{
  \title{\bf Variational inference based on \\[1mm]
a subclass of closed skew normals}
  \author{Linda S. L. Tan  (statsll@nus.edu.sg)\thanks{
    The authors gratefully acknowledge the support of this research by the Ministry of Education, Singapore, under its Academic Research Fund Tier 2 (Award MOE-T2EP20222-0002).}\hspace{.2cm}\\
    Department of Statistics and Data Science, National University of Singapore\\
    and \\
    Aoxiang Chen (e0572388@u.nus.edu) \\
    Department of Statistics and Data Science, National University of Singapore}
  \maketitle
} \fi

\if1\blind
{
  \bigskip
  \bigskip
  \bigskip
  \begin{center}
    {\LARGE\bf Title}
\end{center}
  \medskip
} \fi

\bigskip
\begin{abstract}
Gaussian distributions are widely used in Bayesian variational inference to approximate intractable posterior densities, but the ability to accommodate skewness can improve approximation accuracy significantly, when data or prior information is scarce. We study the properties of a subclass of closed skew normals constructed using affine transformation of independent standardized univariate skew normals as the variational density, and illustrate how it provides increased flexibility and accuracy in approximating the joint posterior in various applications, by overcoming limitations in existing skew normal variational approximations. The evidence lower bound is optimized using stochastic gradient ascent, where analytic natural gradient updates are derived. We also demonstrate how problems in maximum likelihood estimation of skew normal parameters occur similarly in stochastic variational inference, and can be resolved using the centered parametrization. Supplemental materials are available online.
\end{abstract}

\noindent%
\vfill

\newpage
\spacingset{1.75} 

\section{Introduction}
Variational inference \citep{Blei2017, Zhang2018} is a popular and scalable alternative to Markov chain Monte Carlo (MCMC) methods for approximate Bayesian inference. Given observed data $y$, the intractable posterior density $p(\theta|y)$ of the variables $\theta$ is approximated by a variational density $q(\theta)$, which is assumed to satisfy some restrictions, such as lying in a parametric family, or being of a factorized form \cite[variational Bayes,][]{Attias1999}. The Kullback-Leibler (KL) divergence between the true posterior and variational density is then minimized under these constraints. As 
\begin{equation} \label{KL div}
\log p(y) = \underbrace{\int q(\theta) \log \frac{p(y, \theta)}{q(\theta)} d\theta}_{\text{Evidence lower bound}} +  \underbrace{\int q(\theta) \log \frac{q(\theta)}{p(\theta|y)} d\theta}_{\text{KL divergence}},
\end{equation}
this is equivalent to maximizing an evidence lower bound on the log marginal likelihood. 

Variational Bayes is widely applied due to the availability of analytic coordinate ascent updates for conditionally conjugate models \citep{Durante2019, Ray2022}, and scalability to massive data sets via subsampling \citep{Hoffman2013}. \cite{Wang2019} established frequentist consistency and asymptotic normality of variational Bayes, but assuming independence among strongly correlated variables can cause underestimation of posterior variance \citep{Turner2011}. Dependencies among variables can be restored through structured \citep{Salimans2013, Hoffman2015} and copula variational inference \citep{Han2016}. \cite{Saha2020} achieved a tighter lower bound by minimizing the R\'{e}nyi-$\alpha$ divergence instead on a nonparametric manifold, while \cite{Tan2021} developed reparametrized variational Bayes for hierarchical models by transforming local variables to be approximately independent of global variables. A partially factorized approach for high-dimensional Bayesian probit regression \citep{Fasano2022}, allowed global variables to depend on independent latent variables. \cite{Loaiza2022} proposed a hybrid variational approximation for state space models, using a Gaussian copula for global variables and sampling latent variables conditionally via MCMC.

Gaussian variational approximation \citep{Opper2009}, where the posterior density is approximated by a Gaussian, is widely used and captures correlation among variables \citep{Quiroz2023}. It is offered via automatic differentiation variational inference in Stan \citep{Kucukelbir2016}, and transformations can improve the normality of constrained or skewed variables \citep{Yeo2000, Yan2019}. As the number of variational parameters scale quadratically with the dimension, \cite{Tan2018} captured posterior conditional independence via a sparse precision matrix, while \cite{Ong2018} used a factor covariance structure, and \cite{Zhou2021} applied Stiefel and Grassmann manifold constraints. From the Bernstein-von Mises theorem, posteriors in parametric models converge to a Gaussian at the rate of $\mathcal{O}(1/\sqrt{n})$ \citep{Doob1949, vandervaart2000}, but large sample sizes may be required for close resemblance.
 
\cite{Anceshi2023} showed that posteriors of the probit, tobit and multinomial probit models belong to unified skew normals \citep{Arellano2006}. \cite{Durante2023} proved the skewed Bernstein-von Mises theorem, showing that posteriors in regular parametric models converge to the generalized skew normal \citep{Genton2005} at a faster rate of $\mathcal{O}(1/n)$, and obtained a skew-modal approximation. \cite{Ormerod2011}, \cite{Lin2019b} and \cite{Zhou2024} employed the multivariate skew normal \citep{Azzalini1999} as variational density, while \cite{Smith2020} used it to construct implicit copulas. \cite{Ormerod2011} used the BFGS algorithm \citep{Nocedal1999} to maximize the lower bound while \cite{Lin2019b} used stochastic natural gradients, and \cite{Zhou2024} matched key posterior statistic estimates. \cite{Salomone2024} used skew decomposable graphical models \citep{Zareifard2016} for variational inference.

We consider a subclass of closed skew normals \citep[CSNs,][]{Gonzalez2004}, constructed via affine transformations of independent univariate skew normals as variational density. This subclass has closed forms for its moments and marginal densities, and is flexible in approximating joint posteriors, as a bounding line is permitted in each dimension, unlike the skew normal whose tail is bounded by a single line \citep{Sahu2003}. Several novel contributions are made. First, we show that the lower bound of this subclass is stationary when the skewness parameter, $\lambda=0$, and that problems in maximum likelihood estimation of skew normal parameters due to this stationary point \citep{Azzalini1999} persist in variational inference.  Second, we show that a ``centered parametrization" composed of the mean, transformed skewness and decomposition of the covariance matrix, akin to that of \cite{Arellano2008}, can resolve optimization issues due to the stationary point. Third, we derive analytic natural gradients \citep{Amari2016} for improving optimization of the lower bound using stochastic gradient ascent, by considering a data augmentation scheme and a decomposition that ensures positive definiteness of the covariance matrix. Finally, we demonstrate that this subclass exhibits flexibility and improved accuracy in approximating the joint posterior in various applications. 

Research in more expressive variational densities is ongoing actively. Mixture models \citep{Jaakkola1998, Campbell2019, Daudel2021} are very flexible and can capture multimodality, but are expensive computationally and may pose challenges in scalability. In amortized variational inference \citep{Kingma2014, Rezende2014}, an inference function (taken as a deep neural network in variational autoencoders) maps each data point to variational parameters of the corresponding latent variable. This approach is scalable as parameters of the deep neural network are shared across data points, but may be suboptimal to variational Bayes \citep{Margossian2024}. Normalizing flows \citep{Papamakarios2021} mold a base density into a more expressive one by applying a sequence of bijective and differentiable transformations. Efficiency of computing the inverse transformation and Jacobian determinant is application dependent, and examples include the planar and radial \citep{Rezende2015}, real-valued non-volume preserving \citep[real NVP,][]{Dinh2017} and neural spline \citep{Durkan2019} flows.

First, we review existing multivariate skew normals and discuss properties of the proposed CSN subclass in Section \ref{Sec_skew normal}. In Section \ref{Sec optimization}, we illustrate and address challenges in optimizing the lower bound when it can be evaluated almost exactly. We then discuss the intractable lower bound and present natural gradients for stochastic inference in Section \ref{Sec gradients}. Performance of the CSN subclass is evaluated across various applications and compared to normalizing flows in Section \ref{Sec: Applications}. Finally, Section \ref{Sec conclusion} concludes with a discussion.

\section{Multivariate skew normal distributions} \label{Sec_skew normal}
The multivariate skew normal \citep[SN,][]{Azzalini1999} is well-studied and the probability density function (pdf) of $\theta \sim \SN_d(\mu, \Sigma, \lambda)$ is
\begin{equation} \label{SN}
p(\theta) = 2 \, \phi_d(\theta|\mu, \Sigma) \, \Phi \{\lambda^\top (\theta - \mu) \},
\end{equation}
where $\mu \in \mathbb{R}^d$ and $\lambda \in \mathbb{R}^d$ are the location and shape parameters respectively, and $\Sigma$ is a $d \times d$ symmetric positive definite matrix. Let $\phi_d(\cdot| \mu, \Sigma)$ denote the $d$-dimensional Gaussian density with mean $\mu$ and covariance matrix $\Sigma$, and $\Phi(\cdot)$ the cumulative distribution function (cdf) of the univariate standard normal. We have $\E(\theta) = \mu + b\delta_s$ and $\Var(\theta) = \Sigma - b^2 \delta_s \delta_s^\top$, where $b = \sqrt{2/\pi}$ and $\delta_s = \Sigma \lambda/\sqrt{1 + \lambda^\top \Sigma \lambda}$. If
\[
\begin{bmatrix} \theta_0 \\  \theta_1 \end{bmatrix} \sim \N(0, \Sigma_s), \quad
\Sigma_s = \begin{bmatrix} 1 &\delta_s^\top \\  \delta_s &\Sigma \end{bmatrix}, \quad 
\theta_0 \in \mathbb{R}, \; 
\theta_1 \in \mathbb{R}^d,
\]
then $\theta$ can be represented as $\theta = \theta_1|\theta_0 > 0$ \citep{Arellano2006}. To analyze the role of $\lambda = (\lambda_1, \dots, \lambda_d)^\top$, we set $\mu = 0$ and $\Sigma = I_d$. In Figure \ref{F1}, the univariate densities become increasingly skewed as $|\lambda|$ increases. For contours plots of the bivariate densities, the angle of inclination is captured by the ratio of $\lambda_1$ and $\lambda_2$ while the degree of flattening against the bounding line increases with their magnitude. 
\begin{figure}[htb!]
\centering
\includegraphics[width=\textwidth]{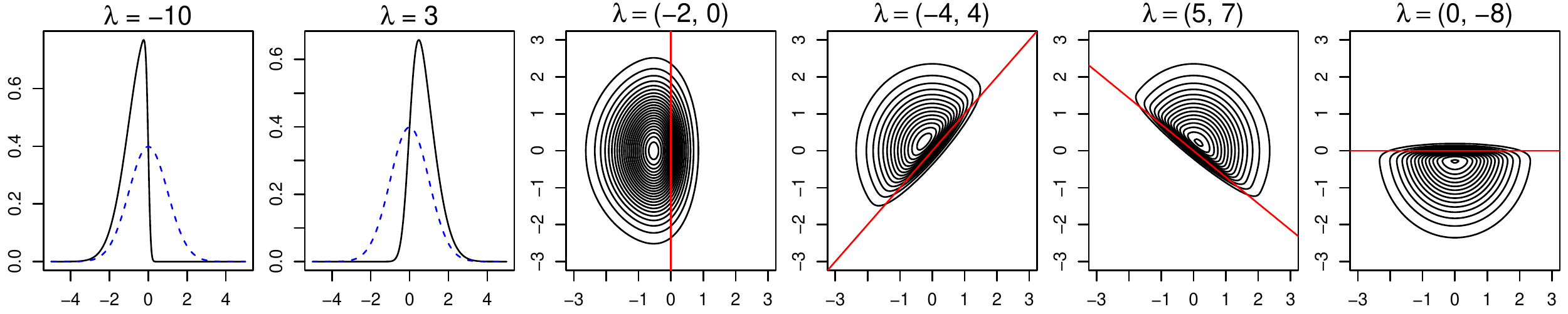}
\caption{First two plots show univariate SN densities. Standard normal is added in dashed lines. Last four are contour plots of bivariate SN densities (bounding line has gradient $-\lambda_1/\lambda_2$).}\label{F1}
\end{figure}
Any inclination angle can be achieved by varying $\lambda$, but densities bounded in more than one direction cannot be captured. This feature of the SN becomes more limiting as the dimension increases, and impedes its ability to provide accurate variational approximations. This constraint stems from $\theta$ being constructed as $\theta_1$ conditioned on a single random variable $\theta_0$ being positive, and \cite{Sahu2003} propose conditioning on $d$ random variables instead.

More generally, \cite{Gonzalez2004} introduced the CSN, which conditions on $q$ random variables. The pdf of $\theta \sim \CSN_{d, q}(\mu, \Sigma, D, \nu, \Delta)$ is
\[
p(\theta) = \phi_d(\theta| \mu, \Sigma) \, \Phi_q(D(\theta-\mu)|\nu, \Delta) / \Phi_q(0| \nu, \Delta + D \Sigma D^\top),
\]
where $\mu \in \mathbb{R}^d$, $\nu \in \mathbb{R}^q$, $D$ is a $q \times d$ matrix, $\Sigma$ and $\Delta$ are $d \times d$ and $q \times q$ positive definite matrices respectively, and $\Phi_q(\cdot|\mu, \Sigma)$ denotes cdf of a normal random vector in $\mathbb{R}^q$ with mean $\mu$ and covariance matrix $\Sigma$. As its name suggests, the CSN enjoys many closure properties, made possible by inclusion of $\Phi_q(\cdot)$ and extra parameters $\nu$ and $\Delta$. For instance, it is closed under the sum of independent CSN random vectors (a property absent for the skew normal) and affine transformation. From Proposition 2.3.1 of \cite{Gonzalez2004}, if $A$ is a $p \times d$ matrix of rank $p \leq d$ and $c \in \mathbb{R}^p$, then 
\begin{equation} \label{affine transf}
A \theta + c \sim \CSN_{p,q}( A \mu + c, \, \Sigma_A, \, D_A, \, \nu, \, \Delta_A),
\end{equation}
where $\Sigma_A = A \Sigma A^\top$, $D_A = D \Sigma A^\top \Sigma_A^{-1}$ and $\Delta_A = \Delta + D \Sigma D^\top - D \Sigma A^\top \Sigma_A^{-1} A \Sigma D^\top $. If $p=d$ and $A$ is invertible, then $D_A = DA^{-1}$ and $\Delta_A = \Delta$. \cite{Arellano2006} proposed an alternative formulation of the CSN known as the ``unified skew normal".

The CSN is attractive for its flexibility and closure properties, but evaluation of $\Phi_q(\cdot)$ is challenging for large $q$ \citep{Genton2018}. For a fast and flexible variational approximation, we consider CSNs where $q=d$, and both $\Delta$ and $\Delta + D\Sigma D^\top$ are diagonal matrices. This setting maintains flexibility, as it can capture densities with a bounding line in each dimension. It can also be optimized efficiently since $\Phi_q(\cdot|\mu, \Sigma)$ can be computed as a product of $q$ univariate cdfs when $\Sigma$ is diagonal. Motivated by \cite{Urbina2022}, we construct a CSN subclass via transformation of independent univariate skew normals for the variational approximation in the next section.

\subsection{A closed skew normal subclass}

For $i=1, \dots, d$, let $v_i \sim \SN_1(0, 1, \lambda_i)$ independently. Then the pdf of $v=(v_1, \dots, v_d)^\top$ is 
\[
p(v) = 2^d \, \phi_d(v|0, I_d) \, \Phi_d(D_\lambda v),
\]
where $D_\lambda = \diag(\lambda)$ and $\Phi_d(x) = \Phi_d(x|0, I_d) = \prod_{i=1}^d \Phi(x_i)$ for any $x=(x_1, \dots, x_d)^\top$. Let $\delta = (\delta_1, \dots, \delta_d)^\top$ and $\tau = (\tau_1, \dots, \tau_d)^\top$ where $\delta_i = \lambda_i/\sqrt{1 + \lambda_i^2}$ and $\tau_i = \sqrt{1-b^2\delta_i^2}$, and $D_\tau = \diag(\tau)$. Then $v \sim \CSN_{d, d}(0, I_d, D_\lambda, 0, I_d)$, where $\E(v) = b\delta$ and $\Var(v) = D_\tau^2$. First we standardize $v$ by defining $z = D_\tau^{-1} (v - b\delta)$, so that $\E(z) = 0$ and $ \Var(z) = I_d$. To construct a CSN random vector $\theta$ with mean $\mu$ and covariance matrix $\Sigma = CC^\top$, we define
\[
\begin{aligned}
\theta &= \mu + C z,
\end{aligned}
\]
where $\mu \in \mathbb{R}^d$ denotes the translation and $C$ is a $d \times d$ invertible matrix representing the linear map. Since the CSN is closed under affine transformation, $\theta \sim \CSN_{d,d}( \mu^* , \Sigma^*, D^*, 0, I_d)$ from \eqref{affine transf}, where $\mu^* = \mu - b C \alpha$, $\Sigma^* = C D_{\tau}^{-2} C^\top$,  $D^* = D_\lambda D_\tau C^{-1}$ and $\alpha = D_\tau^{-1} \delta$. The pdf is 
\begin{equation} \label{Proposed CSN}
q(\theta) = 2^d \, \phi(\theta|\mu^*, \Sigma^*) \, \Phi_d \{D^*(\theta - \mu^*) \},
\end{equation}
and $q(\theta)$ reduces to $\SN_1(\mu - b \sigma \alpha, \sigma^2 / \tau^2, \lambda \tau / \sigma)$, where we write $C$ as $\sigma$ if $d=1$. For $d > 1$, $q(\theta)$ does not coincide with the skew normal. Based on the transformation of $\theta$ from $v$, we can write $q(\theta) = p(v) |C^{-1}| |D_\tau|$. Then 
\begin{equation} \label{log q}
\log q(\theta) = d\log(2) - \frac{d}{2}\log(2\pi) - \frac{v^\top v}{2} - \log |C| + \sum_{i=1}^d \{\log \Phi (\lambda_i v_i) + \log \tau_i\},
\end{equation}
where $v = D_\tau C^{-1} (\theta - \mu) + b \delta$. This CSN subclass falls in the affine independent variational inference framework of \cite{Challis2012}, who used fast Fourier transform to fit variational approximations to a class of nonconjugate models. In contrast, we do not impose any model restrictions and the lower bound is optimized using stochastic (natural) gradient ascent. Our variational approximation is more parsimonious compared to \cite{Smith2020}, who transformed the target distribution elementwise before approximating with Gaussian or skew normal distributions. Another related approach by \cite{Salomone2024} uses a Cholesky decomposition of the precision instead of covariance matrix.

\subsection{Decomposition of covariance matrix} \label{Sec decomposition}

Figure \ref{F2} shows some bivariate densities of the CSN subclass. Compared to Figure \ref{F1}, fan-shaped densities with two bounding lines can now be captured, but the number of bounding lines decreases if $\lambda_1$ or $\lambda_2$ is zero. Figure \ref{F2} also highlights a limitation in the role of $\lambda$: it is not possible to rotate the densities without altering the degree of flattening using $\lambda$. The role of rotation (and scaling) have thus fallen onto the linear map $C$.  

\begin{figure}[htb!]
\centering
\includegraphics[width=\textwidth]{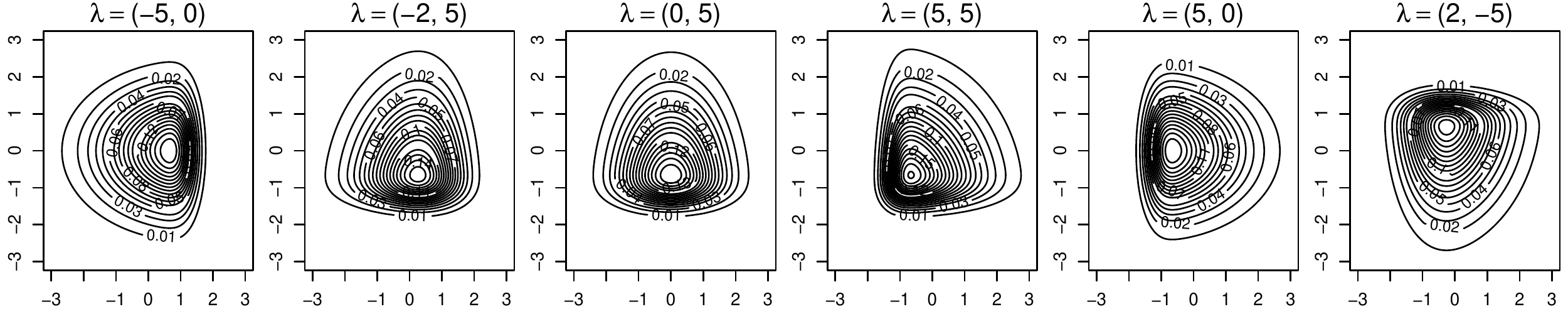}
\caption{Bivariate density contour plots of CSN subclass with $\mu=0$ and $C=I_2$.}\label{F2}
\end{figure}

For decomposing the covariance matrix $\Sigma = CC^\top$, ensuring $\Sigma$ is positive definite while allowing unconstrained optimization of the lower bound is crucial. Two main approaches are Cholesky factorization and spectral decomposition \citep{Pinheiro1996}. Cholesky factorization is more efficient computationally and independence assumptions can be imposed easily. If $\theta = (\theta_1^\top, \theta_2^\top)^\top$ and $C= \diag(C_{11}, C_{22})$ is a corresponding partitioning, then $\theta_1$ and $\theta_2$ are independent from \eqref{Proposed CSN}. Conditional independence can also be imposed via Cholesky decomposition of $\Sigma^{-1}$ \citep{Salomone2024}. However, a disadvantage is that $C$ will be a triangular matrix that is unable to capture rotations in every plane. 

We also considered parametrizing in terms of Givens rotation matrices $\{G_{ij}\}$. As $\Sigma$ is positive definite, we can write $\Sigma = Q D Q^\top$, where $D$ is a diagonal matrix containing eigenvalues of $\Sigma$, and $Q$ is an orthogonal matrix containing orthonormal eigenvectors of $\Sigma$. We represent $Q$ as $\prod_{i=1}^{d-1} \prod_{j=i+1}^d G_{ij}$, where $G_{ij}$ is an identity matrix with 4 elements modified as $G_{i,i} = G_{j,j} = \cos(\omega_{ij})$ and $G_{j,i} = - G_{i,j} = \sin(\omega_{ij})$. The number of parameters of $C = Q D^{1/2}$ remains as $d(d+1)/2$, but all scalings are captured by $D^{1/2}$, which must be performed {\em prior} to rotations by $Q$. As order matters, this limits permissible transformations. Moreover, further computation is required in finding $C$ given $\omega$, although $Q^{-1} = Q^\top$.

As $CC^\top$ is already symmetric, we consider the $C=LU$ decomposition, where $L$ is lower triangular and $U$ is upper triangular with unit diagonal for uniqueness. This ensures that $\Sigma$ is positive definite and $C^{-1}$ can be computed efficiently. In the following, we consider $C$ as a Cholesky factor (CSNC) or $C=LU$ (CSNLU). In high-dimensional settings, decomposing the covariance matrix into a factor structure instead, can enhance computational efficiency and provide insight into parameter interactions \citep{Ong2018}.

\subsection{Properties of closed skew normal subclass} \label{sec prop of CSN}
Next, we discuss some properties of this CSN subclass. Let $C[i,:]$ and $C[:,j]$ denote the $i$th row and $j$th column of a matrix $C$ respectively, $C_{ij}$ denote the $(i,j)$ element and define $\zeta_r(x)$ as the $r$th derivative of $\zeta_0(x) = \log \{ 2\Phi(x) \}$. 

\begin{enumerate}[(P1)]
\item $\E(\theta) = \mu$ and $\Var(\theta) = \Sigma = CC^\top$ (by construction).
\item The cumulant generating function (log of the moment generating function) is 
\[
K_\theta(t) = t^\top \mu^* + \frac{1}{2} t^\top \Sigma^* t  + \sum_{j=1}^d \zeta_0(\alpha_j C[:,j]^\top t), \quad t \in \mathbb{R}^d.
\]
This result follows from Lemma 2.2.2 of \cite{Gonzalez2004}.
\item Marginal density of the $i$th element of $\theta$ is 
$\theta_i \sim \CSN_{1,d} (\mu^*_i  , \Sigma^*_{ii}, D_i, 0, \Delta_i )$,
where $\mu^*_i = \mu_i - b C[i, :]^\top \alpha$ is the $i$th element of $\mu^*$, $\Sigma^*_{ii} = b_i^\top b_i$ is the $(i, i)$ element of $\Sigma^*$, $D_i = D_\lambda b_i /\Sigma^*_{ii}$, $\Delta_i = I_d + D_\lambda ( I_d - b_i b_i^\top/\Sigma^*_{ii}) D_\lambda$ and $b_i = D_\tau^{-1} C[i,:]$. This result is obtained by setting $A$ as $e_i^\top$ and $c=0$ in \eqref{affine transf}, where $e_i \in \mathbb{R}^d$ is binary with only the $i$th element equal to one.
Hence, the marginal pdf of $\theta_i$ is 
\[
q(\theta_i) = 2^d \,  \phi(\theta_i|\mu^*_i, \Sigma^*_{ii}) \, \Phi_d(D_i(\theta_i - \mu_i^*)|0, \Delta_i) ,
\]
which depends only on $\mu_i$, $C[i,:]$ and $\lambda$. To evaluate $\Phi_d(\cdot)$ for large $d$, it may be useful to use the hierarchical Cholesky factorization approach \citep{Genton2018}.
\item Cumulant generating function of $\theta_i$ is $K_{\theta_i}(t) = \mu_i^* t + \tfrac{1}{2} \Sigma_{ii}^* t^2 + \sum_{j=1}^d  \zeta_0(\alpha_j C_{ij} t)$, $t \in \mathbb{R}$.
\item The marginal mean, variance and Pearson's index of skewness of $\theta_i$ are 
$\mu_i$, $\Sigma_{ii}$ and $b(2b^2-1) (\sum_{j=1}^d \alpha_j^3 C_{ij}^3)/ \Sigma_{ii}^{3/2}$ respectively.
\item To simulate from $q(\theta)$, we employ the stochastic representation, 
\begin{equation*}
\begin{aligned}
\theta|w &\sim \N( \mu  + C D_\alpha \widetilde{w} ,\,  C D_{\kappa}^2 C^\top), 
\quad 
w \sim \N(0, I_d),
\end{aligned}
\end{equation*}
where $\widetilde{w} =|w| - b \bfone$, $|w| = (|w_1|, \dots, |w_d|)^\top$, $D_\alpha = \diag(\alpha)$, $D_{\kappa} = \diag(\kappa)$ and $\kappa =  (\kappa_1, \dots, \kappa_d)^\top$ where $\kappa_i = 1/\sqrt{1 + (1-b^2) \lambda_i^2}$. This result arises from representation of $v_i \sim \SN_1(0,1,\lambda_i)$ as $v_i = |v_{0}|\delta_i + \sqrt{1- \delta_i^2} v_{1}$ where $v_{0}, v_{1} \sim \N(0,1)$ independently \citep{Arellano2006}. To simulate from $q(\theta)$, we generate $w_1$ and $w_2$ independently from $\N(0, I_d)$ and set $\theta = C(D_\kappa w_2 + D_\alpha \widetilde{w}_1) + \mu$, where $\widetilde{w}_1 = |w_1| - b \bfone$.
\end{enumerate}

\section{Optimization of evidence lower bound} \label{Sec optimization}
Let $p(y|\theta)$ denote the likelihood of observed data $y$ given unknown model parameters $\theta$ and $p(\theta)$ be a prior on $\theta$. Suppose we wish to approximate the posterior density $p(\theta|y)$ by a variational density $q(\theta)$ with parameters $\eta$. From \eqref{KL div}, optimal variational parameters are found by maximizing the lower bound $\mL = \E_q \{h(\theta)\}$, where $\E_q$ denotes expectation with respect to $q(\theta)$, and $h(\theta)=  \log p(y, \theta) - \log q(\theta)$. Applying chain rule, 
\begin{equation} \label{score fn grad}
\begin{aligned}
\nabla_\eta \mL 
&= \int h(\theta)  \nabla_\eta q(\theta) \d\theta - \int q(\theta) \nabla_\eta \log q(\theta) \d\theta 
= \E_q \{h(\theta)  \nabla_\eta \log q(\theta)\},
\end{aligned}
\end{equation}
since $ \nabla_\eta q(\theta) =  q(\theta)  \nabla_\eta \log q_\eta (\theta)$ and the second term (expectation of the score) is zero. This gradient expression is known as the {\em score function method}. Optimal variational parameters can be found by searching for the stationary points of $\mL$ where $\nabla_\eta \mL = 0$, but this approach may encounter difficulties when the skew normal is used as variational density.

\subsection{Stationary point and alternate parametrizations}
The skew normal reduces to a Gaussian at $\lambda = 0$, and is known for peculiar traits in this vicinity. The log-likelihood is non-quadratic in shape, and has a stationary point at $\lambda = 0$ for any observed data, creating difficulties in maximum likelihood estimation. The Fisher information matrix is also singular at $\lambda=0$ due to a mismatch between the symmetric kernel and skewing function \citep{Hallin2014}. This results in slower convergence, and possibly a bimodal asymptotic distribution for maximum likelihood estimates. 

Here, we consider a CSN subclass which reduces to the SN when $d=1$, but our goal is to maximize the lower bound, and it is unclear if peculiar behavior around $\lambda = 0$ will persist. Theorem \ref{Thm stat pt}, with proof in the supplement S1, shows that the lower bound also has a stationary point at $\lambda =0$, which is unlikely to be the global maximum unless the posterior density does not exhibit any skewness and is best approximated by a Gaussian.

\begin{theorem} \label{Thm stat pt}
Let $\hat{\mu}$ and $\hat{C}$ be optimal parameters of a Gaussian variational  approximation, $q_G(\theta) = \phi_d(\theta|\mu, \Sigma)$, of the posterior density where $\Sigma = CC^\top$. That is, the lower bound is maximized at $\mu=\hat{\mu}$ and $C=\hat{C}$. If a density $q(\theta)$ with parameters $\mu$, $C$ and $\lambda$ from the SN in \eqref{SN} where $\Sigma = CC^\top$, or CSN subclass in \eqref{Proposed CSN} is used as the variational approximation instead, then the lower bound will have a stationary point at $\mu=\hat{\mu}$, $C=\hat{C}$ and $\lambda=0$. 
\end{theorem}

This phenomenon creates optimization difficulties as iterates may get stuck at the stationary point. Gradient ascent algorithms face sensitivity to initializations and slow convergence near $\lambda = 0$ due to surface flatness. To resolve this issue, we consider alternate parametrizations. \cite{Arellano2008} showed that the Fisher information is nonsingular and log-likelihood is more quadratic-like for a centered parametrization based on its mean, covariance and per dimension skewness index. For our CSN subclass, the mean is $\mu$, covariance is $CC^\top$, while Pearson's index of skewness from (P5) is $b(2b^2-1) \alpha^3$ in one dimension. Hence we consider $\alpha^3 =(\alpha_1^3, \dots, \alpha_d^3)^\top$ in place of $\lambda$, where each $\alpha_i^3$ lies between $\pm (1-b^2)^{-3/2}$ ($\approx \pm 4.565$). In this article, scalar functions applied to vector arguments are evaluated elementwise. Another option to overcome singularities in the Fisher information is to replace $\lambda$ by $\lambda^3$ \citep{Hallin2014}. We compare these parametrizations later.

\subsection{Expression of evidence lower bound}
In some problems, the lower bound $\mL = \E_q\{ \log p(y, \theta) - \log q(\theta) \}$ may be evaluated in closed form or efficiently using numerical integration. For the CSN subclass, the entropy,
\[
\begin{aligned}
H_q = -\E_q \{\log q(\theta)\} &= \tfrac{d}{2} \{ \log(\pi/2) + 1\} + \log |C| - \sum_{i=1}^d [2\E_{\phi(u|0,1)} \{\Phi (\lambda_i u) \log \Phi (\lambda_i u) \} + \log{\tau_i} ],
\end{aligned}
\]
which is obtained by taking expectation of \eqref{log q}. Hence $H_q$ is a function of $\lambda$ and $C$ only, and it is symmetric about $\lambda = 0$. In addition, $\lambda  = 0$ is a stationary point of $H_q$ since
\[
\nabla_{\lambda_i} H_q = b^2 \kappa_i^2 \lambda_i / (1+ \lambda_i^2) -  b  \E_{\phi(u|0,(1+\lambda_i^2)^{-1})} \{u \log \Phi(\lambda_i u) \}/ \sqrt{1 + \lambda_i^2}
\]
is zero at $\lambda=0$. Plots of the entropy in Figure \ref{F3} reveal a flat region around $\lambda = 0$, spanning from about $-1$ to 1 if we parametrize in terms of $\lambda$, complicating the identification of an optimal $\lambda$ without strong data or prior influence. The situation improves if $\lambda^3$ is used but there are many sharp corners, while $\alpha^3$ yields contours that are almost quadratic in shape. 

The term $\E_q\{ \log p(y, \theta) \}$ is model dependent and may not be available in closed form. Such cases are discussed in Section \ref{Sec gradients}. Lemma \ref{Lem closed form exp for CSN} is useful for evaluating the expectation of terms of the form $\exp(s^\top \theta)$ in $\log p(y, \theta)$, and its proof is given in the supplement S2.
\begin{lemma} \label{Lem closed form exp for CSN}
For the density $q(\theta)$ in the CSN subclass in \eqref{Proposed CSN} and any $s \in \mathbb{R}^d$, 
$\exp(s^\top \theta) q(\theta)  = M \tilde{q}(\theta)$, where $M= 2^d  \Phi( D_\alpha C^\top s)  \exp(s^\top  \mu^* + s^\top \Sigma^* s/2)$, $\tilde{q}(\theta)$ is the pdf of $\CSN_{d,d}(\mu^* + \Sigma^*s,  \Sigma^*,  D^*, - D_\lambda D_\tau^{-1} C^\top s, I_d)$, whose mean is $\tilde{\mu} = \mu^* + \Sigma^*s + C D_\alpha \zeta_1(D_\alpha C^\top s)$ and covariance is $\tilde{\Sigma} = \Sigma^* + C D_\alpha^2 \diag\{ \zeta_2(D_\alpha C^\top s) \} C^\top$.
\end{lemma}

\begin{figure}[tb!]
\centering
\includegraphics[width=0.85\textwidth]{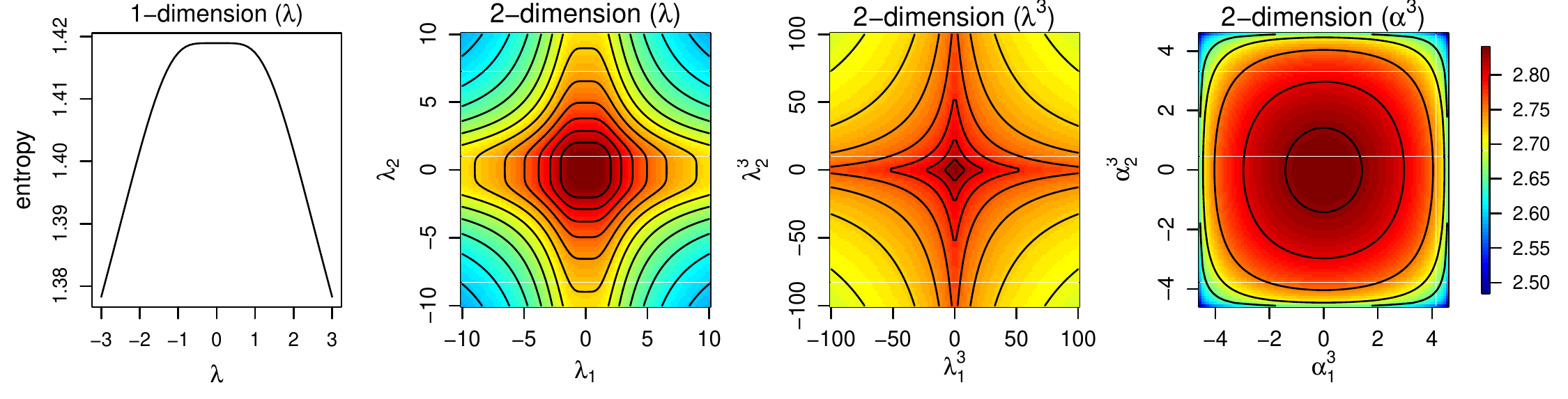}
\caption{Plots of the entropy for different parametrizations with $C=I_2$.}\label{F3}
\end{figure}

Next, we discuss some applications where the lower bound $\mL$ can be computed (almost) exactly to examine how different parametrizations alter its surface. First, we explain how the accuracy of different approaches are assessed.

\subsection{Accuracy assessment}
In variational inference, a higher value of $\mL$ indicates an approximation closer in KL divergence to the true posterior, but it is hard to quantify how significant an improvement of say 0.1 is. Following \cite{Faes2011}, we assess the accuracy of $q(\theta)$ using the {\em integrated absolute error}, $\text{IAE}(q) = \int |q(\theta) -q_{\text{GS}}(\theta) |d\theta$, by comparing it with a gold standard $q_{\text{GS}}(\theta)$, computed using numerical integration or the kernel density estimate based on MCMC samples. The IAE is invariant to monotone transformations of $\theta$ and lies in (0,2). Thus we use $\text{Accuracy}(q) = 1 - \text{IAE}(q) /2$, expressed as a percentage. MCMC is performed in {RStan}, where 2 chains each of 50000 iterations are run in parallel and the first half is discarded as burn-in. Remaining 50000 draws are used for the kernel density estimate.

In Section \ref{Sec: Applications}, we further evaluate the {\em multivariate accuracy} of variational approximation relative to MCMC using two metrics, cross-match non-bipartite statistic \citep{Rosenbaum2005} and maximum mean discrepancy \citep[MMD, ][]{Gretton2012}. Following \cite{Yu2023}, we generate 1000 samples each from the variational approximation and MCMC. Optimal non-bipartite pairings (NBPs) are derived from the pooled ranked samples, and instances of cross-match NBPs containing one sample from each distribution are counted. A higher count suggests greater similarity between the distributions. Following \cite{Zhou2023}, we compute $M^*=-\log (\max(\text{MMD}, 0) + 10^{-5})$,  where
\begin{equation*}
\text{MMD} = \frac{1}{m(m-1)}\sum_{i \neq j}^m[k(\textbf{x}_v^{(i)},\textbf{x}_v^{(j)})+k(\textbf{x}_g^{(i)},\textbf{x}_g^{(j)})-k(\textbf{x}_v^{(i)},\textbf{x}_g^{(j)})-k(\textbf{x}_v^{(j)},\textbf{x}_g^{(i)})],
\end{equation*}
$k$ is the radial basis kernel function,  $\textbf{x}_v^{(1)},\dots,\textbf{x}_v^{(m)}$ and $\textbf{x}_g^{(1)},\dots,\textbf{x}_g^{(m)}$ are independent samples from variational approximation and MCMC respectively, and $m = 1000$. A higher $M^*$ indicates higher multivariate accuracy. Evaluation of each metric is repeated 50 times.

\subsection{Normal sample} \label{sec normal sample}
Let $\theta=(\theta_1,\theta_2)$ and consider the model $y_i|\theta \sim \N(\theta_1, \exp(\theta_2))$ for $i=1, \dots, n$, with priors $\theta_1 \sim \N(0, \sigma_0^2)$ and $\exp(\theta_2) \sim \IG (a_0, b_0)$, where $a_0 = b_0 = 0.01$ and $\sigma_0^2 = 10^4$ \citep{Ormerod2011}. Derivation of $\E_q \{\log p(y, \theta)\}$ using Lemma \ref{Lem closed form exp for CSN} is given in the supplement S3.

First, let $\theta_1 = 0$ and $\theta = \theta_2$ only. The true posterior of $\theta$ is $\IG(a_0 + n/2, b_0 + \sum_{i=1}^n y_i^2/2)$. Following \cite{Ormerod2011}, we simulate $n=6$ observations by setting $\exp(\theta_2) = 225$. Writing $C$ as $\sigma$, Figure \ref{F4} shows the profile lower bound, $\mL(\lambda) = \mL(\hat{\mu}(\lambda), \hat{\sigma}(\lambda), \lambda)$, where $\hat{\mu}$ and $\hat{\sigma}$ maximize $\mL$ for any given $\lambda$. The leftmost plot is similar to the profile log likelihood in \cite{Arellano2008}, which has a stationary point at $\lambda=0$ and a non-quadratic shape that increases the risk of getting stuck at $\lambda = 0$ during optimization. 
\begin{figure}[tb!]
\centering
\includegraphics[width=0.85\textwidth]{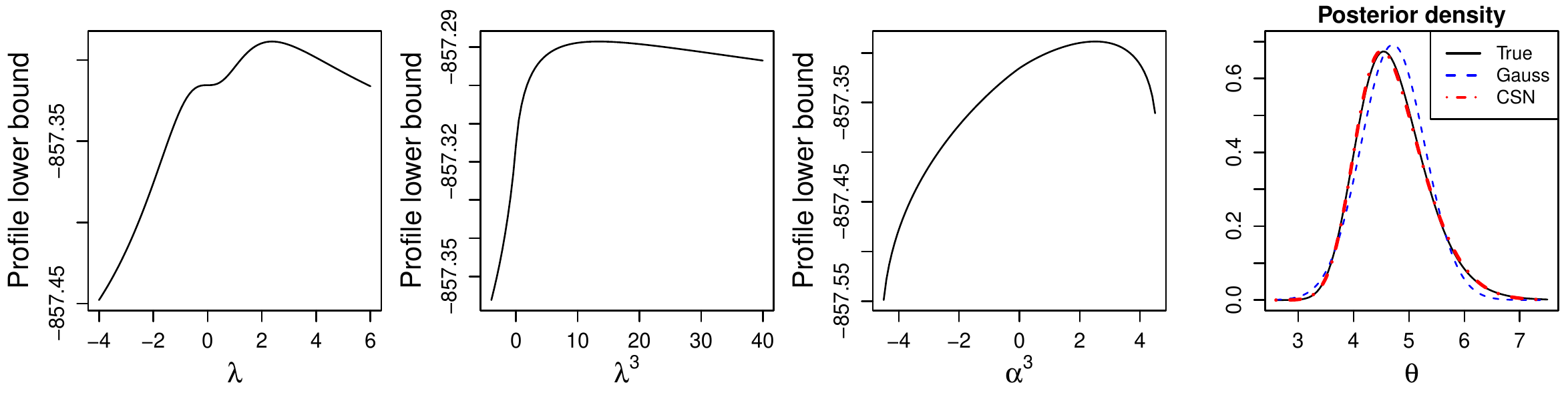}
\caption{First three plots show the profile lower bound for $\lambda$, $\lambda^3$ and $\alpha^3$. Last plot shows true posterior density, and the Gaussian and CSN variational approximations.}\label{F4}
\end{figure}
It shows that the ``centered" parameters $\mu$ and $\sigma$ alone cannot eliminate the stationary point, but replacing $\lambda$ by $\lambda^3$ or $\alpha^3$ achieves this goal, and $\alpha^3$ yields a more pronounced mode. The last plot compares the true posterior with the CSN ($\mL = -26.45$) and Gaussian ($\mL = -26.47$) variational approximations computed using {\tt optim} in R. The CSN has a higher accuracy (99.0\%) than the Gaussian (92.6\%).

\begin{figure}[tb!]
\centering
\includegraphics[width=\textwidth]{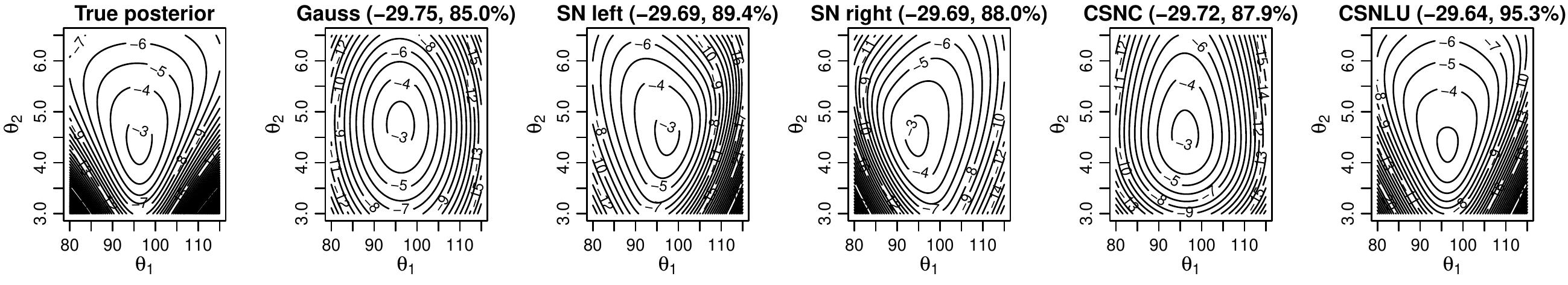}
\caption{Contour plots of true posterior and variational approximations. For SN, results from 2 initializations capturing left or right tail are shown. $\mL$ and accuracy are reported in ().}\label{F5}
\end{figure}
When $\theta_1$ and $\theta_2$ are unknown, $n=6$ observations are simulated with $\theta_1 = 100$ and $\exp(\theta_2) = 225$. Normalizing constant of the true posterior is estimated by integrating with respect to $\theta_1$ analytically and $\theta_2$ numerically. Figure \ref{F5} compares contour plots of the true posterior with variational approximations computed using BFGS via {\tt Optim} in {\tt R}. SN and CSN are more sensitive to initialization than the Gaussian and we report the best of different initializations. The Gaussian is least accurate as it cannot capture skewness. SN can capture the left or right tail depending on the initialization, but not both as it is bounded in only one direction. CSNC does not have the correct orientation due to its inability to perform rotations,  but still improves on the Gaussian. CSNLU achieves the highest lower bound and accuracy, albeit at the cost of more parameters. Figure \ref{F6} shows the lower bound of CSNLU for different parametrizations, and the stationary point at zero can be observed for $\lambda$. Again, $\alpha^3$ yields a quadratic-like surface that facilitates optimization. A further example on the Poisson generalized linear model is given in the supplement S4. 

\begin{figure}[tb!]
\centering
\includegraphics[width=0.85\textwidth]{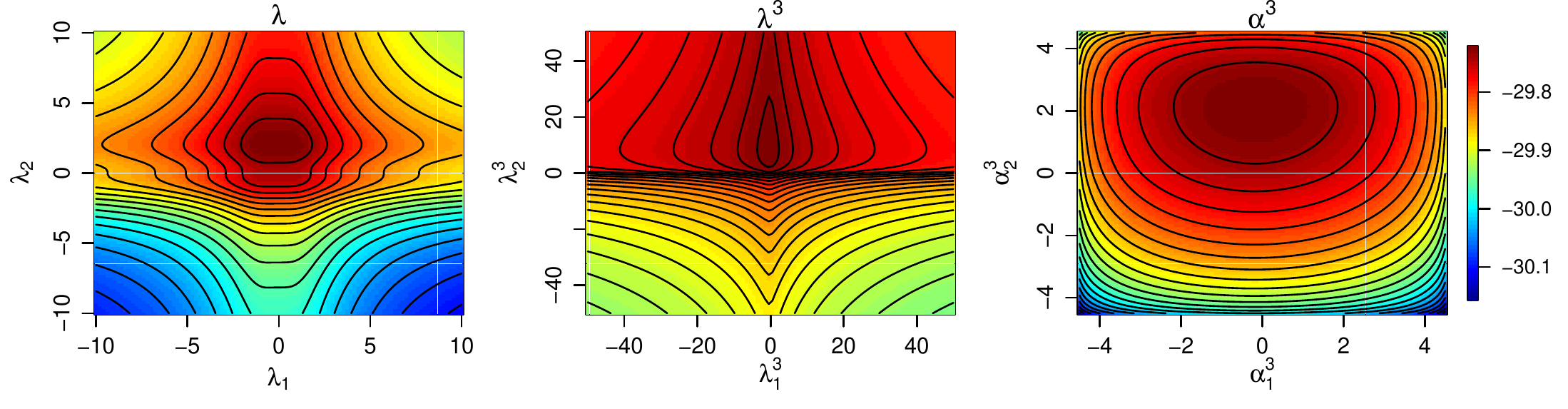}
\caption{Lower bound of CSNLU at $\mu=\hat{\mu}$, $L=\hat{C}$ and $U=I_d$ for different parametrizations of $\lambda$, where $(\hat{\mu}, \hat{C})$ are optimal parameters for Gaussian approximation.}\label{F6}
\end{figure}

\section{Stochastic variational inference for CSN subclass} \label{Sec gradients}
For models where $\E_q\{\log p(y, \theta)\}$ is not analytically tractable, we can maximize $\mL$ with respect to $\eta$ using stochastic gradient ascent \citep{Robbins1951}, where an update 
\[
\eta_{t+1} = \eta_t + \rho_t \widehat{\nabla}_\eta \mL(\eta_t)
\]
is made at the $t$th iteration and $\widehat{\nabla}_\eta \mL(\eta_t)$ is an unbiased estimate of the true gradient $\nabla_\eta \mL$ at $\eta_t$. This algorithm will converge to a local optimum under certain regularity conditions, if the stepsize $\rho_t$ satisfies $\sum_{t=1}^\infty \rho_t = \infty$ and $\sum_{t=1}^\infty \rho_t < \infty$ \citep{Spall2003}. 

Unbiased gradient estimates can be computed via the reparametrization trick \citep{Kingma2014}. Suppose $\theta = \mathcal{T}_\eta(z)$ where $\mathcal{T}_\eta(\cdot)$ is a differentiable function containing all information involving $\eta$, and $z$ can be sampled from a density $p(z)$ independent of $\eta$. Then $\nabla_\eta \mL = \E_{p(z)} \{ \nabla_\eta \theta \nabla_\theta h(\theta) \}$, and an unbiased estimate is $\widehat{\nabla}_\eta \mL = \nabla_\eta \theta \nabla_\theta h(\theta)$, where $\theta = \mathcal{T}_\eta(z)$ and $z \sim p(z)$. The score function method in \eqref{score fn grad} also yields an unbiased estimate of $\nabla_\eta \mL$, but the reparametrization trick often returns estimates with lower marginal variance, such as for quadratic log likelihoods in variational Bayes \citep{Xu2019}. In Gaussian variational approximation \citep{Titsias2014}, if $\theta \sim \N(\mu, CC^\top)$, then $\theta = Cz + \mu$ where $z \sim \N(0, I_d)$. For the CSN subclass, we can use (P6) to find $\nabla_\eta \theta$. For $\nabla_\theta h(\theta) = \nabla_\theta \log p(y, \theta) - \nabla_\theta \log q(\theta)$, $\nabla_\theta \log p(y, \theta)$ is model dependent and
\[
\begin{aligned}
\nabla_\theta \log q(\theta) &= C^{-T} D_{\tau} \{ D_\lambda \zeta_1(D_\lambda v) - v \},
\end{aligned}
\]
where $v = D_\tau C^{-1} (\theta - \mu) + b \delta$. If $C$ is the Cholesky factor, then $\eta = (\mu^\top, \lambda^\top, \vech(C)^\top)^\top$, where $\vech(C)$ vectorizes lower triangular elements in $C$ columnwise from left to right. Let $w_3 = \{ \widetilde{w}_1 - (1-b^2) D_\lambda w_2 \} \odot C^\top \nabla_\theta h(\theta)$, where $\odot$ is the Hadamard product. Then 
\[
\begin{aligned}
\widehat{\nabla}_\eta \mL = \nabla_\eta \theta \nabla_\theta h(\theta)
= \begin{bmatrix}
\nabla_\theta h(\theta) ^\top, & (D_\kappa^3 w_3)^\top, & \vech(\nabla_\theta h(\theta) z^\top)^\top
\end{bmatrix}^\top.
\end{aligned}
\]
If $C  = L U$, then $\eta = (\mu^\top, \lambda^\top, \vech(L)^\top,  \vech_u(U)^\top)^\top$, where $\vech_u(\cdot)$ vectorizes elements above the diagonal columnwise from left to right. $\widehat{\nabla}_\mu \mL$ and $\widehat{\nabla}_\lambda \mL$ remain unchanged while
\[
\begin{aligned}
\widehat{\nabla}_{\vech(L)} \mL = \vech(\nabla_\theta h(\theta) z^\top U^\top), \quad  \text{and} \quad
\widehat{\nabla}_{\vech_u(U)} \mL  = \vech_u(L^\top \nabla_\theta h(\theta) z^\top).
\end{aligned}
\]
More details are given in the supplement S5. If $\alpha^3$ is used as parameter instead of $\lambda$, then 
\[
\widehat{\nabla}_{\alpha^3} \mL = (\nabla_{\alpha^3} \lambda) \widehat{\nabla}_\lambda \mL = \tfrac{1}{3}D_\alpha^{-2} D_\kappa^{-3}(D_\kappa^3 w_3) = \tfrac{1}{3}D_\alpha^{-2} w_3.
\]

BFGS cannot be applied directly when $\E_q\{\log p(y, \theta)\}$ is not analytically tractable, but the stochastic optimization problem can be converted into a deterministic one by using sample average approximation \citep{Burroni2023}, thus enabling the use of BFGS.

\subsection{Natural gradients for CSN subclass} \label{sec Nat grad}
The steepest ascent direction is given by the (Euclidean) gradients derived previously when distance between $\eta$s is measured by the Euclidean metric. However, as we are optimizing $\mL$ with respect to $\eta$ in a curved parameter space, the Euclidean metric may not be appropriate. When distance between densities is measured using KL divergence, the steepest ascent direction is given by the natural gradient \citep{Amari2016}, which premultiplies the Euclidean gradient with the inverse of the Fisher information matrix, $\mathcal{I}_\theta (\eta) = - \E_q\{\nabla_\eta^2 \log q(\theta)\}$. 

For the SN, the Fisher information matrix is singular at $\lambda = 0$ and involves expectations which are not analytically tractable \citep{Arellano2008}, complicating the use of natural gradients. \cite{Lin2019b} overcame this using a minimal conditional exponential families representation, but their $\Sigma$ update does not ensure positive definiteness. We derive analytic natural gradients for the CSN subclass by considering the Cholesky or LU decomposition, and the Fisher information of $q(\theta, w)$. From (P6),
\[
\begin{aligned}
\log q(\theta,w) &= \log q(\theta|w) + \log q(w) \\
&= - d\log(2\pi) - \log|C|- \sum_{i=1}^d \log \kappa_i  - \frac{w^\top w}{2}  -\frac{z^\top D_\kappa^{-2} z}{2}  - \frac{\widetilde{w}^\top D_{\lambda}^2\widetilde{w}}{2}  +z^\top D_{\lambda}D_\kappa^{-1} \widetilde{w},
\end{aligned}
\]
where $z=C^{-1} (\theta - \mu)$. If $C = LU$, then $\log|C| = \log|L|$ since $U$ has unit diagonal. This approach is feasible as $\mL = \int q(\theta) h(\theta) d\theta = \int \int q(\theta, w) h(\theta) d\theta dw$, and $q(\theta, w)$ shares the same parameter $\eta$ as $q(\theta)$. Hence we can optimize $\mL$ in the parameter space of $q(\theta, w)$. The natural gradient is $\widetilde{\nabla}_\eta \mL = \mathcal{I}_{\theta, w} (\eta)^{-1} \nabla_\eta \mL$ instead of $\mathcal{I}_{\theta} (\eta)^{-1} \nabla_\eta \mL$, where 
\begin{equation} \label{tvedef}
\begin{aligned}
\mathcal{I}_{\theta, w}(\eta) 
&= - \E_{q(\theta,w)} \{\nabla_\eta^2 \log q(\theta,w) \}
&= \mathcal{I}_\theta (\eta) + \E_{q(\theta)} \{\mathcal{I}_{w|\theta} (\eta|\theta)\}.
\end{aligned}
\end{equation}
Unlike $\mathcal{I}_\theta (\eta)$ which has singularities, $\mathcal{I}_{\theta, w}(\eta)$ is positive definite, and we can derive its inverse and the natural gradients analytically. Hence it is useful to consider natural gradients based on $\mathcal{I}_{\theta, w}(\eta)$, although progress in gradient ascent may be more conservative due to \eqref{tvedef}. The results are presented in the next two theorems, whose proofs are in the supplement S6. 

First, we define some notation. For any square matrix $X$, let $X_\ell$ and $X_u$ be the lower and upper triangular matrices derived from $X$  respectively by replacing all elements above the diagonal by zeros, and all elements on and below the diagonal by zeros. For $s \in \mathbb{R}^{d(d+1)/2}$, $\vech^{-1}(s)$ is a $d \times d$ lower triangular matrix filled with elements of $s$ columnwise from left to right. For $t \in \mathbb{R}^{d(d-1)/2}$,  $\vech_u^{-1}(t)$ is a $d \times d$ upper triangular matrix with zero diagonal filled with elements of $t$ columnwise from left to right. Let $K =\kappa^2 \bfone^\top$.

\begin{theorem} \label{thm1}
For the variational density $q(\theta)$ in the CSN subclass in \eqref{Proposed CSN}, if $C$ is the  Cholesky factor of $\Sigma$ and $\eta = (\mu^\top, \lambda^\top, \vech(C)^\top)^\top$, then the natural gradient is 
\[
\widetilde{\nabla}_\eta \mL = \begin{bmatrix}
(C D_\kappa^2 C^\top) \nabla_ \mu \mL \\
\tfrac{1}{(1-b^2) (2\kappa^2 - \kappa^4)} \odot \nabla_\lambda \mL + \tfrac{\lambda}{2 - \kappa^2} \odot \diag( A_1)\\
 \vech( CA_1 ),
\end{bmatrix},
\]
where $G= \{C^\top \vech^{-1}(\nabla_{\vech(C)} \mL)\}_\ell$ and $A_1 = \diag(\tfrac{\alpha\kappa}{2} \odot \nabla_\lambda \mL) +  G \odot  \{ K - \diag(\tfrac{\kappa^4}{2}) \}$. 
\end{theorem}

\begin{theorem} \label{thm2}
For the variational density $q(\theta)$ in the CSN subclass in \eqref{Proposed CSN}, if $C = LU$ where $L$ is a lower triangular matrix and $U$ is an unit diagonal upper triangular matrix, and $\eta = (\mu^\top, \lambda^\top, \vech(L)^\top,  \vech_u(U)^\top)^\top$, then the natural gradient is given 
\[
\begin{aligned}
\widetilde{\nabla}_\mu \mL &= (C D_\kappa^2 C^\top) \nabla_ \mu \mL, \\
\widetilde{\nabla}_{\vech(L)} \mL &= \vech (L \mG_\ell ), 
\end{aligned}\quad
\begin{aligned}
\widetilde{\nabla}_\lambda \mL &= \tfrac{1}{(1-b^2)(2\kappa^2 - \kappa^4)} \odot  \nabla_\lambda \mL + \tfrac{\lambda}{2-\kappa^2} \odot \diag(H ),  \\
\widetilde{\nabla}_{\vech_u(U)} \mL &= \vech_u [ U \{ K_u \odot (F - H^\top )\} + \mG_u U ],
\end{aligned}
\]
where $a = \vech \{\diag(\frac{1}{2-\kappa^2}) + (1/K)_\ell - (K_u)^\top\}$, 
\[
\begin{aligned}
G &= \{L^\top \vech^{-1}(\nabla_{\vech(L)} \mL) \}_\ell, \\
F &= \{U^\top \vech_u^{-1} (\nabla_{\vech_u(U)} \mL )\}_u,
\end{aligned} \quad
\begin{aligned}
H &= A_2 \odot [ U^\top \{ G - (U^{-T} F  U^\top)_\ell  \}U^{-T} \\
& \;+ \diag (\tfrac{\lambda}{2-\kappa^2} \odot  \nabla_\lambda \mL)  - ( K \odot F)^\top] ,  
\end{aligned}
\quad 
\begin{aligned}
\mG &=U H  U^{-1}, \\
A_2 &= \vech^{-1}(1/a).
\end{aligned}
\]
\end{theorem}

If $\nabla_\eta \mL$ is intractable, it is replaced by $\widehat{\nabla}_\eta \mL$ in Theorems 2 and 3 to obtain unbiased natural gradient estimates. If we parametrize in terms of $\alpha^3$ instead, then $
\widetilde{\nabla}_{\alpha^3} \mL = (\nabla_\lambda \alpha^3)^\top \widetilde{\nabla}_\lambda \mL = 3D_\alpha^2 D_\kappa^3 \widetilde{\nabla}_\lambda \mL$ from \cite{Tan2024a}.

\section{Applications} \label{Sec: Applications}
This section investigates the performance of proposed methods using logistic regression, survival models,  zero-inflated negative binomial models and generalized linear mixed models (GLMMs). An independent $\N(0, \sigma_0^2)$ prior is specified for each element in $\theta$, where $\sigma_0^2=100$. Using the Gaussian as a baseline, we compare CSN variational approximations with planar and real NVP flows, which are elaborated in the supplement S7. For each application, $\log p(y, \theta)$ and its gradient are given in the supplement S8. 

CSN algorithms are initialized using Gaussian variational approximations and run for 50,000 iterations, with $\lambda$ initialized as $\bfone$ or $-\bfone$. Lower bound estimates are averaged every 1000 iterations to reduce noise. We report results from the $\alpha^3$ parametrization, which is less sensitive to initialization, unless stated otherwise. For Euclidean gradients, Adam \citep{Kingma2015} is used to compute the stepsize, while a constant stepsize is used for natural gradients. As a sign-based variance adapted approach (Balles and Hennig, 2018), Adam is useful for Euclidean gradients, but not natural gradients as scaling by the Fisher information is neglected \citep{Tan2024a}. 

Normalizing flows are trained on GPU using the PyTorch package {\tt normflows} \citep{Stimper2023} for 20,000 iterations, with $\N(0, I_d)$ as base density. The transforming function used in planar flow is $\tanh(\cdot)$. In real NVP, the scale and translation functions are multilayer perceptrons, each with a single hidden layer containing $2d$ units and an ReLU activation function. A binary mask is used to update alternate elements in $\theta$, in turn between layers. Unless stated otherwise, the flow length $K=8$, and the number of samples for Monte Carlo gradient estimation is the square root of the total number of observations \citep{Caterini2021}. All computations are performed on an Intel Core i9-9900K CPU @ 3.60GHz with 16GB RAM and an Nvidia GeForce GTX 1660 graphics card.

\subsection{Logistic regression model} \label{Sec applications}

\begin{figure}[b!]
\centering
\includegraphics[width=\textwidth]{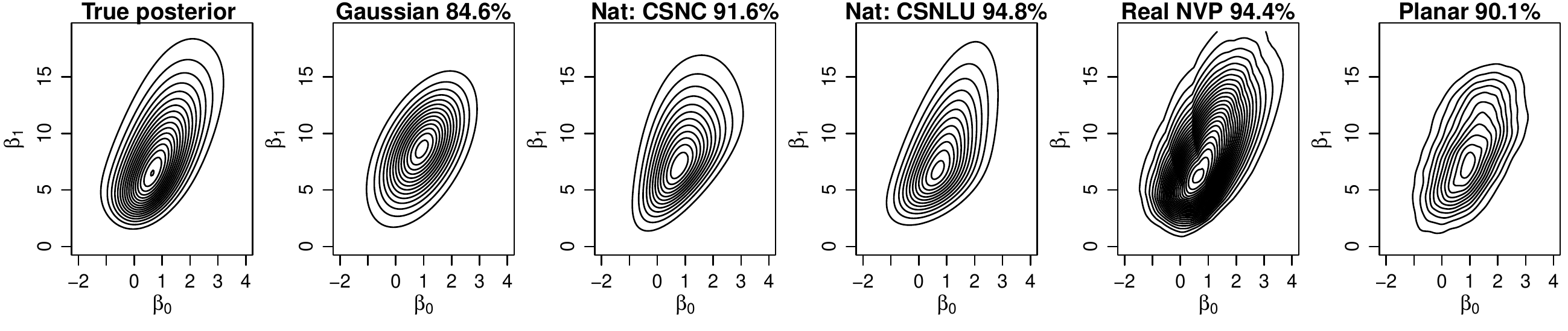}
\caption{Contour plots and accuracies for bioassay data.}\label{Bioplot1}
\end{figure}
In the bioassay data \citep{Racine1986}, a dose ($x_i$) of chemical compound is administered to $n_i$ animals in each of $i=1, \dots, n$ trials, with $n=4$. The number of deaths in each trial is modeled independently as $y_i \sim \text{Binomial}(n_i, p_i)$, and a logistic regression model, $\logit(p_i) = \beta_0 + \beta_1 x_i$, is used so that $\theta = (\beta_0, \beta_1)^\top$. Normalizing constant of the true posterior is estimated using Monte Carlo integration based on one billion samples. 

From Figure \ref{Bioplot1}, CSNLU and real NVP flow have the highest accuracies of 94--95\%. The orientation of CSNC does not match the true posterior well, likely due to its inability to rotate. The first two boxplots of Figure \ref{Bioplot2} indicate that real NVP flow produces an approximation closest to MCMC in terms of NBP, but CSNLU dominates according to MMD. When the CSN algorithms are initialized using $\lambda = \bfone$, difference between the $\lambda$ and $\alpha^3$ parametrizations is negligible, but significant differences can be observed when initialized using $\lambda = -\bfone$. The 3rd and 4th plots of Figure \ref{Bioplot2} show that $\lambda$ converges slowly and is stuck at zero when parametrized using $\lambda$, but the iterates are able to escape the stationary point at zero and converge to a better mode when parametrized using $\alpha^3$. The last plot of Figure \ref{Bioplot2} shows that natural gradients consistently yield higher lower bounds than Euclidean gradients when a constant stepsize of 0.001 is used.  

\begin{figure}[tb!]
\centering
\includegraphics[width=\textwidth]{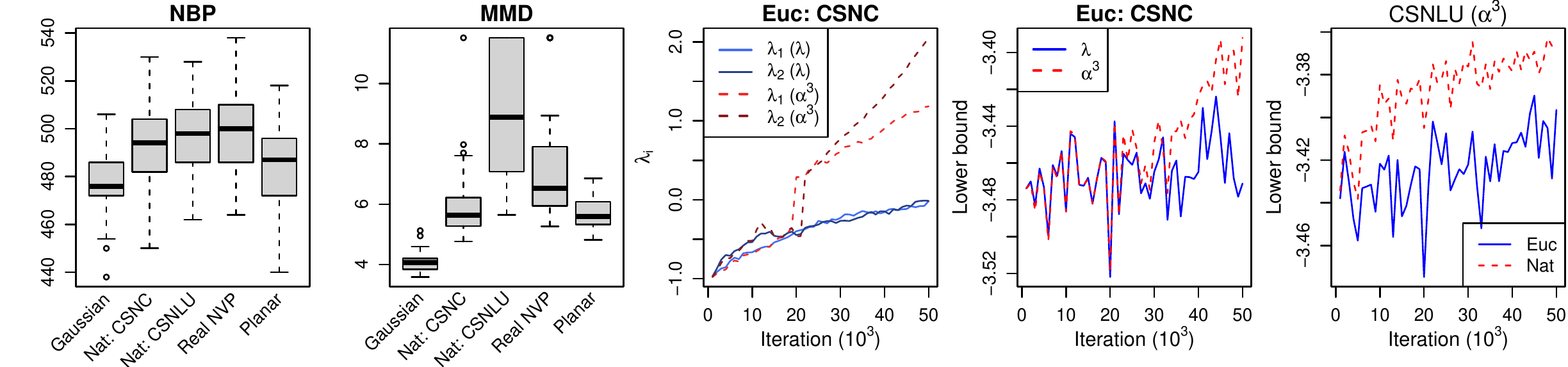}
\caption{Bioassay data. Boxplots of the NBP and MMD metrics, trace of $\lambda$ and lower bound from CSNC, and lower bound of CSNLU.}\label{Bioplot2}
\end{figure}

Consider the logistic regression model for binary responses, where $y_i \sim \text{Bernoulli}(p_i)$ independently and $\logit(p_i) = x_i^\top \theta$  for $i=1, \dots, n$, where $\theta \in \mathbb{R}^d$ and $x_i \in \mathbb{R}^d$ denote the coefficients and covariates for the $i$th response respectively. We fit the model to German credit data from the UCI Machine Learning Repository, featuring $n=1000$ individuals classified by good or bad credit risks. Quantitative predictors are standardized and qualitative predictors are dummy-coded. The ground truth is based on MCMC samples. For CSN, kernel density estimates are computed based on 50,000 samples as it is challenging to evaluate $\Phi_d(\cdot)$ for $d=49$. The boxplots in Figure \ref{germanplot} show that CSNC and CSNLU outperform the Gaussian, real NVP and planar flows consistently across the multivariate metrics (NBP and MMD) and accuracies in marginal density estimates. Using natural gradients, CSNC and CSNLU achieve a high minimum accuracy of 98.3\%. Figure \ref{germanplot} also shows the marginal densities of two coefficients whose Gaussian approximation accuracy is less than 85\%. All the other methods (being able to accommodate skewness) improve on the Gaussian, and are more aligned with MCMC results. 

\begin{figure}[t!]
\centering
\includegraphics[width=\textwidth]{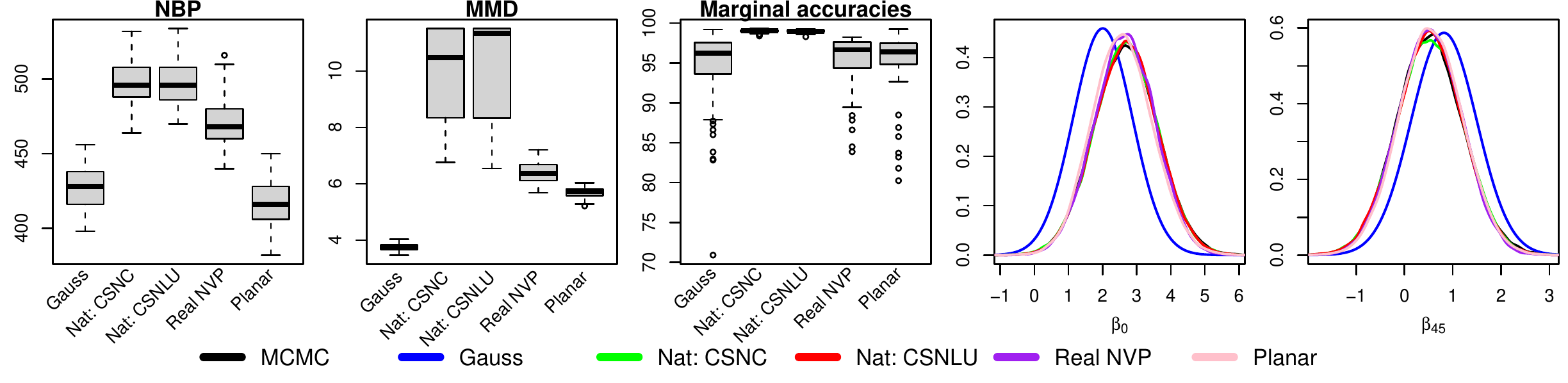}
\caption{ German data. Boxplots of NBP, MMD and marginal accuracies, and marginal density estimates of $\beta_0$ and $\beta_{45}$.}\label{germanplot}
\end{figure}

\subsection{Zero-inflated negative binomial model}
The zero-inflated negative binomial model is used to analyze counts with excessive zeros while allowing for overdispersion. We apply it to a dataset on number of fish ($y_i$) caught by visitor $i$ at a national park in a day for $i=1, \dots, n$, where $n=250$. An observation $y_i$ is 0 with probability $\varphi_i$, or is generated from a negative binomial distribution with probability $1-\varphi_i$. The pdf of the negative binomial distribution with parameters $\alpha$ and $\mu_i$ is 
\[
p(y_i) = \frac{\frac{1}{\alpha}^{1/\alpha} \Gamma(y_i + \frac{1}{\alpha})\mu_i^{y_i}}{y_i!\Gamma(\frac{1}{\alpha}) (\mu_i + \frac{1}{\alpha})^{y_1 + 1/\alpha} }, 
\]
which is obtained by integrating out $\tau_i$ from the hierarchical model, $y_i |\tau_i \sim \Poisson(\mu_i \tau_i)$ and $\tau_i \sim \text{Gamma}(\frac{1}{\alpha}, \frac{1}{\alpha})$. The response $\mu_i$ is modeled as $\log \mu_i = x_i^\top \beta$, where $x_i$ includes an intercept, an indicator for use of live bait ({\tt livebait}) and number of accompanying persons ({\tt persons}). The probability $\varphi_i$ (that the visitor did not fish) is modeled as $\logit (\varphi_i) = z_i^\top \gamma$, where the covariates in $z_i$ include an intercept, number of accompanying children ({\tt child}) and an indicator for camping ({\tt camper}).  Thus $\theta =(\beta^\top, \gamma^\top, \log \alpha)^\top$ and $d=7$.

\begin{table}[tb!]
\centering \small
\begin{tabular}{l|cccccccccc}
\hline
& $\mathcal{L}$ & NBP & MMD & $\beta_0$ & $\beta_1$ & $\beta_2$  & $\gamma_0$ & $\gamma_1$ & $\gamma_2$ & $\log \alpha$ \\ 
\hline
Gaussian & -425.8 & 401.2 & 2.4 & 99.0 & 99.1 & 99.4 & 67.4 & 65.5 & 68.1 & 95.0 \\
  Nat: CSNC & -425.3 & 461.0 & 3.6 & 99.0 & 98.8 & 98.8 & 83.9 & 83.2 & 78.1 & 96.9 \\
  Nat: CSNLU & {\bf -425.2} & {\bf 473.5} & {\bf 3.9} & 99.0 & 99.1 & 99.2 & 84.7 & 85.1 & 85.2 & 96.5 \\
  Real NVP &  -425.9 & 369.6 & 3.0 & 90.4 & 93.9 & 89.1 & 75.3 & 73.8 & 77.3 & 95.8 \\
  Planar &  -426.2 & 330.4 & 2.1 & 78.7 & 81.5 & 88.2 & 64.3 & 62.2 & 62.3 & 92.2 \\
\hline
\end{tabular}
\caption{Average lower bound, NBP, MMD, and marginal accuracies for fish data.}\label{T_fish}
\end{table}

\begin{figure}[b!]
\centering
\includegraphics[width=\textwidth]{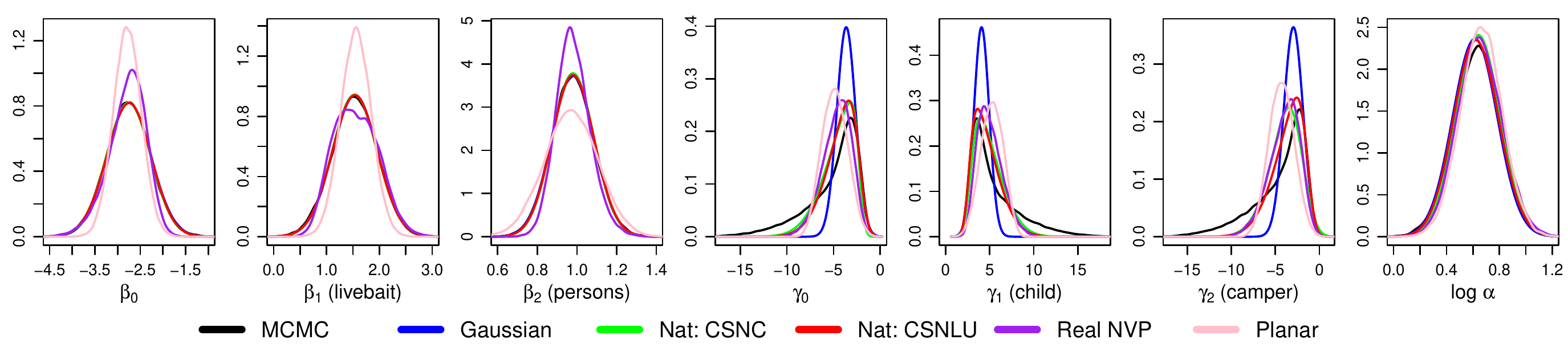}
\caption{Marginal density estimates of coefficients for fish data.}\label{fishplot1}
\end{figure}

As the flow methods have difficulty converging, we reduce the flow lengths to 2, and increase the number of Monte Carlo samples for gradient estimation to 50, to achieve convergence. Table \ref{T_fish} shows that CSNLU provides a fit closest to the true posterior based on the lower bound. It is also most similar to the MCMC kernel estimate in terms of NBP and MMD, while CSNC is second best. For marginal density estimates, the Gaussian is highly accurate for $\beta$ ($\sim$99\%), but less so for $\gamma$ ($\sim$67\%) whose marginal posteriors are highly asymmetrical. From Figure \ref{fishplot1}, CSN improves on the Gaussian, but is unable to capture all the skewness in $\gamma$, and CSNLU ($\sim$85\%) slightly outperforms CSNC ($\sim$83\%). Real NVP ($\sim$75\%) also improves on the Gaussian, but not as well as the CSN, and planar flow has weaker  performance than real NVP. In higher dimensions, the accuracies of the variational approximations are likely lower than in one dimension. Figure \ref{fishplot2} shows that the bivariate marginal posterior of $(\gamma_1, \log \alpha)$ is shaped irregularly. The contour plots of planar flow (62.0\%) and Gaussian (65.1\%), being elliptical, are inadequate, while real NVP (73.1\%), CSNC (81.7\%) and CSNLU (83.4\%) can capture skewness in the tail more effectively. 

\begin{figure}[tb!]
\centering
\includegraphics[width=\textwidth]{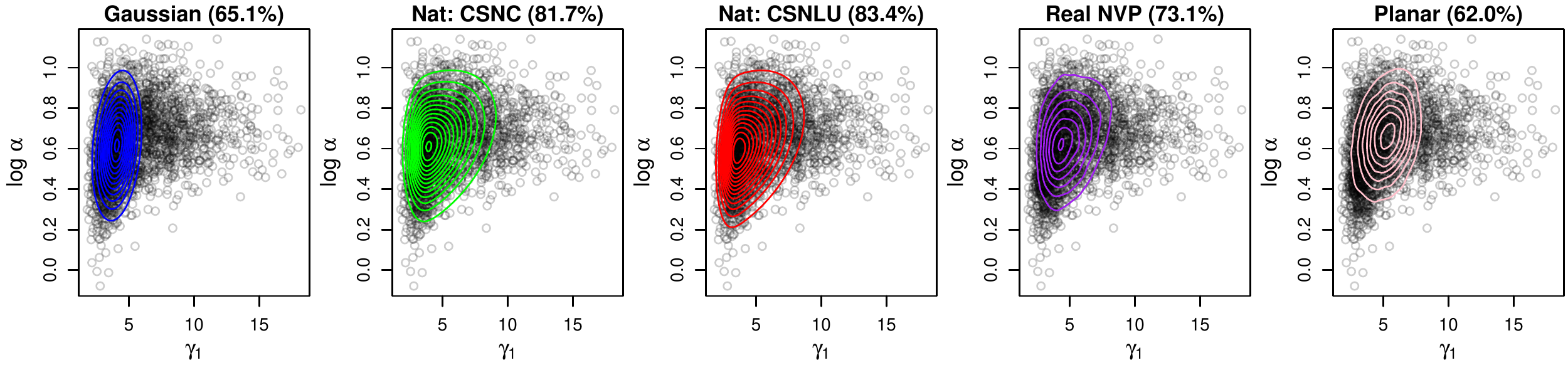}
\caption{Fish data. Bivariate contour plots of variational approximations superimposed on 2500 randomly selected MCMC samples and their accuracies.}\label{fishplot2}
\end{figure}

\subsection{Survival model}
We consider a dataset for analyzing the effects of a hip-protection device, age and sex, on the risk of hip fractures in $n = 148$ patients. The $i$th patient is observed up to a time $t_i$ at which an indicator $d_i$ for whether fracture occurs is recorded. For the Weibull proportional hazards model, the hazard function is $h(t_i) = \rho_i t_i^{\rho_i - 1} \exp(x_i^\top \beta)$  for $i=1, \dots, n$, and the survival function is $S(t_i) = \exp \{-\exp(x_i^\top \beta) t_i^{\rho_i} \}$, where $\rho_i> 0$. The covariates $x_i$ includes an intercept, an indicator for use of hip-protection device ({\tt protect}) and {\tt age}, which is standardized. The hazard curves for men and women are assumed to be different in shape, and the ancillary variable $\rho_i$ is modeled as $\log \rho_i = z_i^\top \gamma$, where $z_i$ includes an intercept and {\tt sex} to account for this difference. Thus $\theta =(\beta^\top, \gamma^\top)^\top$ and $d=5$.

\begin{table}[htb!]
\centering
\begin{small}
\begin{tabular}{lcccccccc}
  \hline
& $\mathcal{L}$ & NBP & MMD & $\beta_0$ & $\beta_1$ & $\beta_2$ & $\gamma_0$ & $\gamma_1$ \\
  \hline
Gaussian & -165.22 & 488.2 & 6.9 & 95.3 & 99.1 & 99.0 & 95.1 & 90.5 \\
  Euc: CSNC & -165.11 & 495.2 & 7.8 & 96.0 & 98.7 & 98.7 & 95.5 & 97.6 \\
  Euc: CSNLU & {\bf-165.07} & {\bf 500.0} & {\bf 9.4} & 98.1 & 98.5 & 98.9 & 99.1 & 97.9 \\
  Real NVP & -165.11 & 493.7 & 8.2 & 98.1 & 96.7 & 97.4 & 97.7 & 98.4 \\
  Planar flow & -165.10 & 491.0 & 7.0 & 95.4 & 98.8 & 98.0 & 97.0 & 93.3 \\
   \hline
\end{tabular}
\caption{Average lower bound, NBP and MMD, and accuracies of marginal densities for hip data.}\label{T_hip}
\end{small}
\end{table}

As natural gradients only provided small improvements in the lower bound, we focus on results obtained using Euclidean gradients. From Table \ref{T_hip}, the CSN and flow methods improve significantly on the Gaussian, and CSNLU produced the best posterior approximation based on the lower bound and the multivariate metrics, NBP and MMD. For marginal density estimates, the Gaussian does very well for $\beta_1$ and $\beta_2$, but is poorer for $\beta_0$, $\gamma_0$ and especially $\gamma_1$. The CSN, real NVP and planar flow (to a smaller degree) are able to improve on these aspects. CSNLU, being more flexible than CSNC, often yields better accuracies, as can be seen in the bivariate marginal posterior plots in Figure \ref{hipplot1}. Similarly, real NVP usually yields better approximations than planar flow, but is more computationally intensive. Figure \ref{hipplot2} shows that when CSNC is initialized from $\lambda=\bfone$, the iterates cannot move past the stationary point at zero under the $\lambda$ parametrization, whereas the $\alpha^3$ parametrization successfully overcomes this issue and achieves a higher lower bound. 

\begin{figure}[htb!]
\centering
\includegraphics[width=\textwidth]{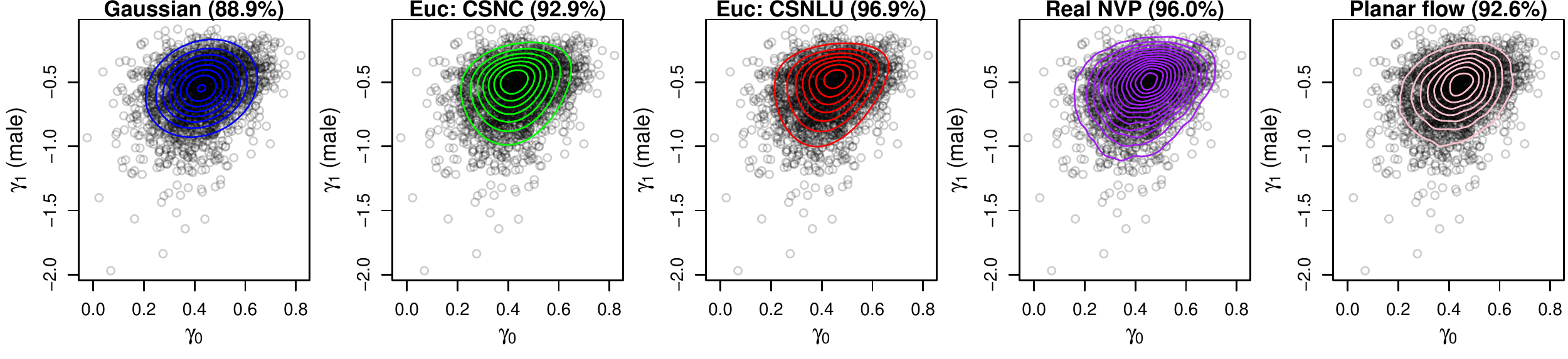}
\caption{Hip data. Bivariate contour plots of variational approximations superimposed on 2500 randomly selected MCMC samples and their accuracies.}\label{hipplot1}
\end{figure}

\begin{figure}[htb!]
\centering
\includegraphics[width=0.85\textwidth]{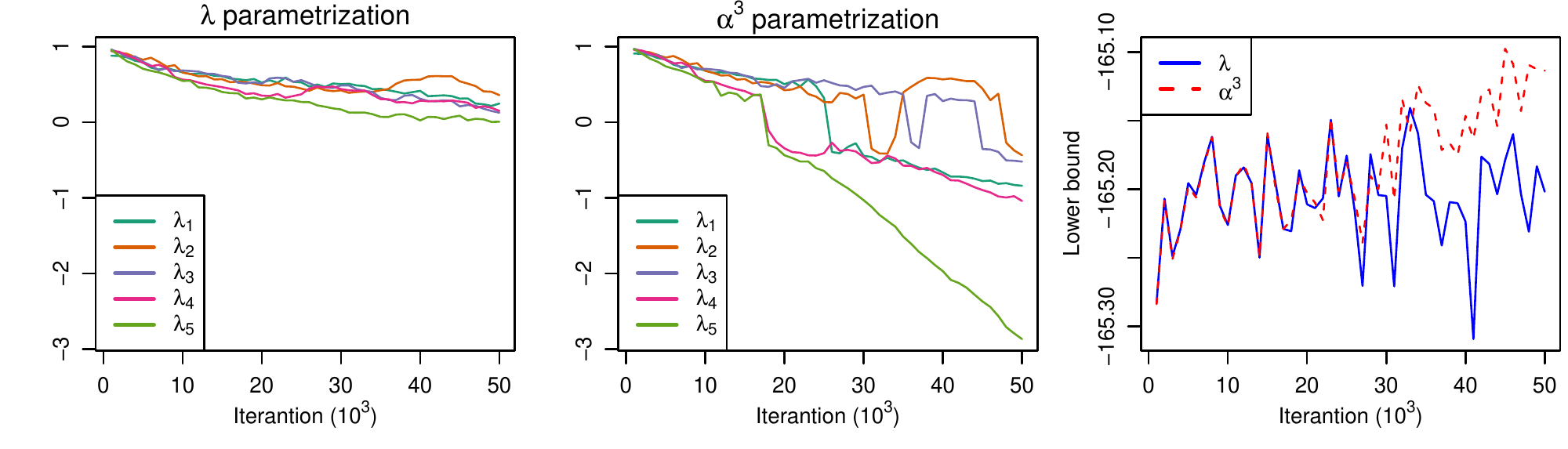}
\caption{Hip data. Trace of $\lambda$ and lower bound from Euc: CSNC.}\label{hipplot2}
\end{figure}

\subsection{Generalized linear mixed model}
Consider a GLMM where $y_i=(y_{i1},\dots,y_{in_i})^\top$ contains $n_i$ observations of the $i$th subject for $i=1, \dots, n$ and  $y=(y_1,\dots,y_n)^{\top}$. Each $y_{ij}$ follows an exponential family distribution with expected value $\mu_{ij}= \E(y_{ij})$. Let $\beta$ denote fixed effects, $b_i \sim \N(0, G^{-1})$ be random effects, and $g(\cdot)$ be a smooth link function. Given covariates $x_{ij}$ and $z_{ij}$,
\begin{equation*}
g (\mu_{ij}) = \eta_{ij} = x_{ij}^\top \beta + z_{ij}^\top b_i \;\;\text{for}\;\; i=1, \dots, n, \;\;j=1, \dots, n_i.
\end{equation*}
The parameters are collected in $\theta = (\theta_L^{\top}, \theta_G^{\top})^{\top}$, where $\theta_G$ denotes the global parameters and $\theta_L = (b_1^{\top},\dots,b_n^{\top})^{\top}$ are the local parameters. More details on model specification are given in the supplement S8. For efficiency, we employ a mean-field variational approximation where $q(\theta) = q(\theta_G) \prod_{i=1}^n q(b_i)$. As computation of NBP is highly intensive in large-scale settings, multivariate performance of GLMMs is assessed using only MMD. 

The polypharm dataset \citep{Hosmer2013} contains 7 binary responses for each of 500 subjects observed for drug usage over seven years. We fit a logistic random intercept model, $\text{logit}(\mu_{ij})= x_{ij}^{\top} \beta + b_i$, to this dataset. The covariates include gender (1 for males, 0 for females), race (0 for white, 1 for all other races), log(age/10), dummy variables for number of outpatient mental health visits (MHV1 $=1$ if 1 to 5, MHV2 $=1$ if 6 to 14, MHV3 $=1$ if $\geq 15$, and 0 otherwise), and a binary indicator for inpatient mental health visits (0 if none, 1 otherwise). As natural gradients did not yield improvements in the lower bound, we report results based on Euclidean gradients. The first two boxplots in Figure \ref{polyplot1} show that CSNC and CSNLU provide approximations of the joint and marginal posteriors that are closest to MCMC, compared to the Gaussian, real NVP and planar flow. The mean MMD of CSNLU (8.36) is also slightly higher than CSNC (8.23).
\begin{figure}[tb!]
\centering
\includegraphics[width=\textwidth]{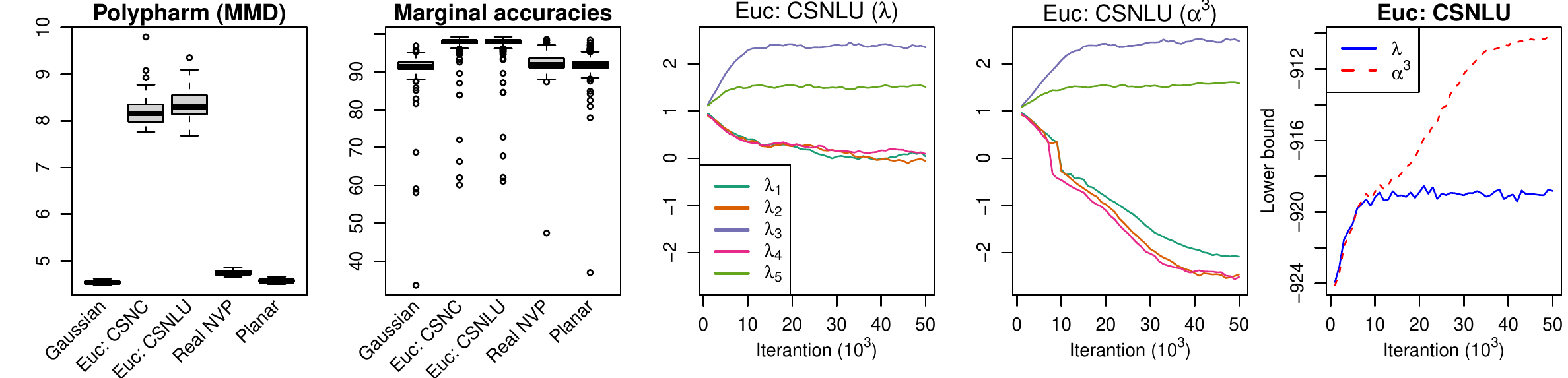}
\caption{Polypharm data. Boxplots of MMD and marginal accuracies, trace of $\lambda$ under the $\lambda$ and $\alpha^3$ parametrizations and trace of lower bound.}\label{polyplot1}
\end{figure}
The next three plots show that iterates of the first five elements of $\lambda$ are unable to traverse the stationary point at zero when CSNLU is initialized from {\bf 1} under the $\lambda$ parametrization. The $\alpha^3$ parametrization resolves this issue and achieves a much higher lower bound. Figure \ref{polyplot2} shows the marginal density estimates of some variables whose Gaussian approximation accuracy is less than 90\%. The CSN is often able to capture the posterior modes and skewness accurately, especially for the random effects. However, posterior variance of the global variables tend to be underestimated, which is likely due to the mean-field assumption. Real NVP provides better estimates of the posterior variance but modal estimates are slightly misaligned.

\begin{figure}[b!]
\centering
\includegraphics[width=0.95\textwidth]{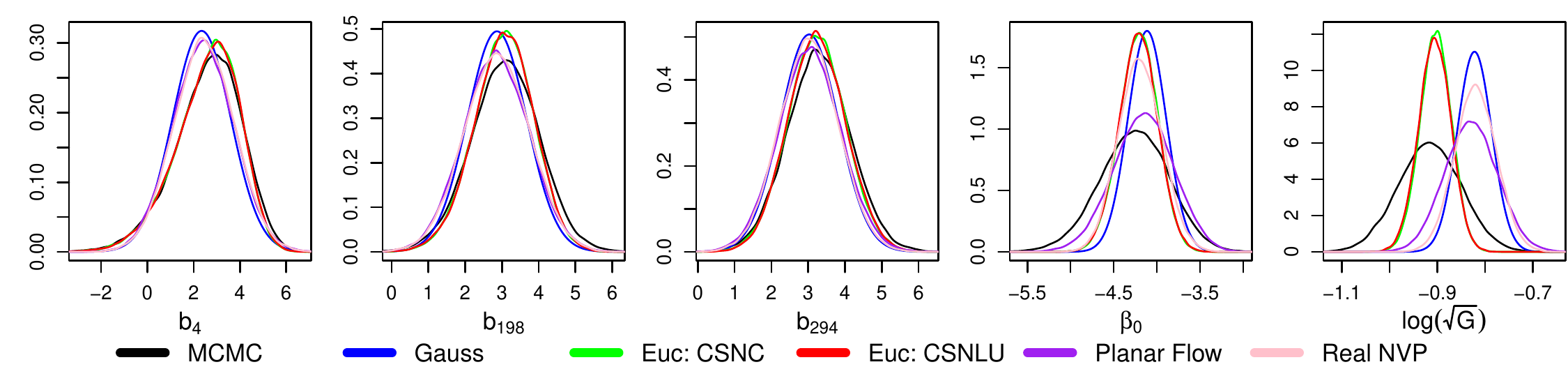}
\caption{Marginal density estimates for some variables in polypharm data.}\label{polyplot2}
\end{figure}

Next, we fit GLMMs to large-scale data through a simulation and real dataset analysis. We generate observations from a random intercept Poisson GLMM with sparse information \citep{Tan2021}, where 
\[
\log (\mu_{ij}) = -2.5 - 2 x_{ij} + b_{i}, \quad 
x_{ij}=(j-4)/10, \quad \text{for} \quad i=1,\cdots,5000,
\quad j=1,\cdots,7.
\]
We also analyze a dataset containing hospital records of diabetics patients from 1999-2008 to study the effect of HbA1c measurements on hospital readmission rates \citep{Strack2014}. A logistic random intercept model is fitted to data from $n = 16,341$ patients with multiple hospital records. The binary response is an indicator for readmission in less than 30 days of discharge, and covariates include {\tt HbA1c} (a variable with 4 levels: HbA1c test not performed, HbA1c performed and in normal range, HbA1c performed and result $> 8$\% with no change in diabetic medication, and HbA1c performed and result $> 8$\% with changes in diabetic medication), {\tt diagnosis} (a variable with 9 levels: Diabetes, Circulatory, Digestive, Genitourinary, Injury, Musculoskelet, Neoplasms, Respiratory and Other), and their interaction. Due to memory constraints, the number of iterations in each of two parallel MCMC chains is reduced to 20,000. After discarding the first half of each chain as burn-in, 20,000 draws are saved for kernel density estimation. For the simulated data, the flow length is reduced to $K=4$ and $K=1$ for planar and real NVP flows respectively to achieve convergence. For the diabetics data, the flow length in planar flow is reduced to $K=2$, while real NVP was not performed due to GPU memory constraints. The number of parameters in planar and real NVP flows scale as $O(Kd)$ and $O(Kd^2)$ respectively, so real NVP is much more computationally intensive than planar flow for large $d$. 

\begin{figure}[b!]
\centering
\includegraphics[width=\textwidth]{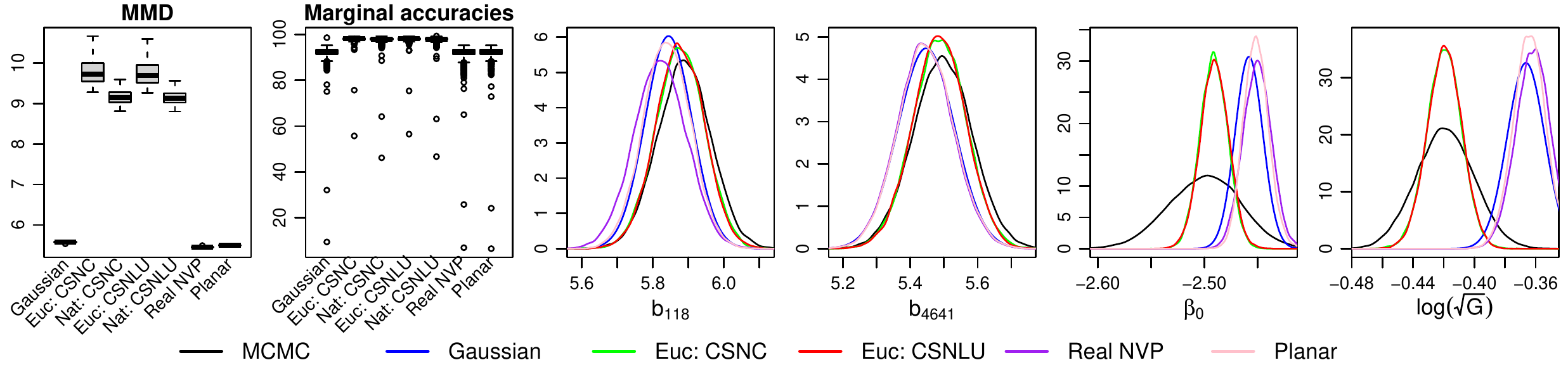}
\caption{Boxplots of MMD, marginal accuracies, and marginal density estimates for some variables in simulated data.}\label{poiplot}
\end{figure}

For both the simulated and diabetics data, use of natural gradients led to higher lower bounds, which indicate a better approximation of the posterior in KLD. The 3rd and 4th plots of Figure \ref{diaplot} also show that natural gradients led to a more steady increase in the lower bound as opposed to the large fluctuations observed when using Euclidean gradients. However, the results based on Euclidean gradients performed better in terms of MMD and marginal accuracies, as can be seen from the boxplots in Figure \ref{poiplot} and \ref{diaplot}. Euc: CSNC has the highest MMD for both datasets, and also performed well in terms of marginal density estimates, achieving high minimum accuracies of 93.8\% and 94.8\% for the simulated and diabetics data respectively, except for one or two outliers. The marginal density estimates of these outliers are shown in Figure \ref{poiplot} and \ref{diaplot}. As for the polypharm data, the CSN can capture the posterior mode and skewness accurately, especially for the random effects, but underestimates the posterior variance for the global parameters. 
 
\begin{figure}[bt!]
\centering
\includegraphics[width=\textwidth]{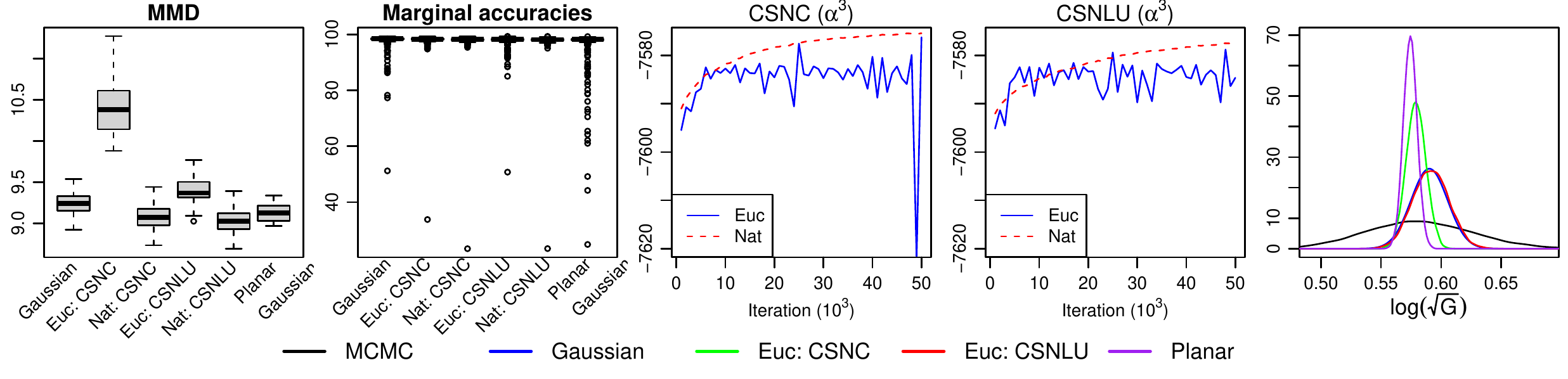}
\caption{Marginal density estimates of coefficients for diabetics data whose Gaussian approximation is less than 90\%.}\label{diaplot}
\end{figure}

\subsection{Computation times}
Table \ref{Ttime} summarizes runtimes for all applications. The Gaussian and CSN algorithms are run in Julia, MCMC in RStan, and normalizing flows in Python on the GPU. Due to the use of different platforms, runtimes may not be directly comparable, but are informative nonetheless. The CSN algorithms provide significant speedup relative to MCMC, and are more efficient and scalable. While CSNC and CSNLU have similar runtimes with Euclidean gradient updates, the use of natural gradients increases their runtimes by a larger margin as the dimension increases. For CSNLU, computation of natural gradients is more intensive due to the inversion of unit upper triangular matrices. CSN algorithms are often much faster than flow-based methods, as they benefit from  the use of analytic gradient updates rather than automatic differentiation. Reducing the flow length and the number of Monte Carlo samples for gradient estimation of normalizing flows can trade-off accuracy for a reduction in computation time. 

\begin{table}[htb!]
\centering
\begin{small}
\begin{tabular}{lrrrrrrrr}
  \hline
 & Gaussian & \thead{Euc: \\CSNC} & \thead{Nat: \\CSNC} &\thead{Euc:\\ CSNLU} & \thead{Nat:\\ CSNLU} & NVP & Planar & MCMC \\
  \hline
Bioassay &     0.1 &     0.9 &     1.6 &     1.0 &     4.4 &   242.5 &   199.0 &    45.8 \\
  German &     2.2 &     5.4 &    11.6 &     9.7 &    60.6 &   245.6 &   217.2 &   511.2 \\
  Fish &     1.1 &     3.2 &     4.0 &     3.3 &     8.1 &   169.6 &   159.4 &   377.9 \\
  Hip &     0.3 &     1.4 &     2.2 &     1.5 &     5.2 &   279.1 &   210.1 &    85.2 \\
  Polypharm &     9.8 &    22.6 &    29.1 &    23.5 &    50.4 &   469.2 &   244.1 &   723.1 \\
  Simulated &    51.0 &   209.1 &   323.0 &   272.1 &   488.9 &  4809.9 &   298.0 &  2664.6 \\
  Diabetics &   254.9 &  1023.3 &  1315.1 &  1292.3 &  1794.0 &     - &  4470.2 & 13462.3 \\
   \hline
\end{tabular}
\caption{Runtimes in seconds for all applications.}\label{Ttime}
\end{small}
\end{table}

\section{Conclusion}\label{Sec conclusion}
In this article, we introduce a subclass of the closed skew normal as an alternative to Gaussian variational approximation that is able to accommodate skewness and is flexible in that a bounding line is permitted in each dimension unlike the original skew normal. This subclass is constructed using affine transformations, and we highlight the limitations instilled by constraining the linear map to be a lower triangular matrix. An LU decomposition is proposed for the linear map when it is a full matrix to ensure ease in inversion during optimization. For the original skew normal, the presence of a stationary point when the skewness is zero is known to create issues in maximum likelihood estimation. We prove that such a stationary point similarly exists in maximization of the evidence lower bound in variational inference, which creates problems in stochastic gradient ascent algorithms. We also demonstrate that parametrizing in terms of $\alpha^3$ is effective in resolving these issues. Finally, we derive analytic natural gradients for maximizing the lower bound using stochastic gradient ascent by considering the augmentation $q(\theta, w)$ instead of $q(\theta)$, that also ensures positive definiteness by using the Cholesky factorization or LU decomposition. The performance of proposed methods is investigated using a variety of statistical applications and comparisons with normalizing flows are provided. 

Flow-based methods are highly flexible in modeling diverse distribution characteristics, such as skewness and multimodality. The adjustable flow lengths also allow for a balance between approximation accuracy and computational efficiency. However, normalizing flows seem to be less effective in capturing skewness in particular compared to the CSN, and may require high-performance GPUs for computation involving large-scale problems. A limitation of our current implementation of CSN approximation for large-scale data is the employment of the mean-field scheme, which reduces its accuracy when complex dependencies among local and global variables exist. While sparse covariance or precision structures can be easily integrated into CSNC, imposing sparsity into CSNLU is more challenging because unlike Cholesky decomposition, the LU decomposition does not inherently preserve sparsity structures. A possible way of overcoming this restriction is to apply reparametrized variational Bayes \cite{Tan2021}, which we are keen to investigate further in future work.

\section{Acknowledgments}
We wish to thank the editor, associate editor and referees for their constructive comments and suggestions, which have improved the manuscript.


%
%

\setcounter{section}{0} \renewcommand{\thesection}{S\arabic{section}}
\setcounter{figure}{0} \renewcommand{\thefigure}{S\arabic{figure}}
\setcounter{table}{0} \renewcommand{\thetable}{S\arabic{table}}
\setcounter{equation}{0} \renewcommand{\theequation}{S\arabic{equation}}
\setcounter{lemma}{0} \renewcommand{\thelemma}{S\arabic{lemma}}

\newpage

\bigskip
\begin{center}
{\large\bf Supplementary material for ``Variational inference \\[1mm]
based on a subclass of closed skew normals"} \\ [3mm]
Linda S. L. Tan (statsll@nus.edu.sg) and Aoxiang Chen(e0572388@u.nus.edu) \\ [1mm]
Department of Statistics and Data Science \\ [1mm]
National University of Singapore
\end{center}

\spacingset{1.5} 

\section{Proof of Theorem 1}

\begin{proof}
We present the proof for the case where $C$ is a full $d \times d$ matrix. If $C$ is lower triangular, the proof is similar with some minor modifications. For the Gaussian variational  approximation $q_G(\theta) $, let $\eta_G = (\mu^T, \vec(C)^T)^T$ denote its parameter, where $\vec(\cdot)$ is vectorization of a matrix columnwise from left to right. We have $\nabla_{\eta_G} \log q_G (\theta) = (\nabla_\mu \log q_G (\theta)^T, \nabla_{\vec(C)} \log q_G (\theta)^T)^T$, where 
\[
\begin{aligned}
\nabla_\mu \log q_G(\theta)  = \Sigma^{-1} (\theta - \mu), \quad 
\nabla_{\vec(C)} \log q_G(\theta) = \vec [\{\Sigma^{-1}(\theta - \mu) (\theta - \mu)^T - I_d \}C^{-T}].
\end{aligned}
\]
Let $h_G(\theta) = \log p(y, \theta) - \log q_G(\theta)$. From \eqref{score fn grad}, gradient of the lower bound $\mL_G$  of $q_G(\theta)$ is
\begin{equation} \label{identity}
\nabla_{\eta_G} \mL_G = \E_{q_G} \{h_G(\theta) \nabla_{\eta_G} \log q_G (\theta) \},
\end{equation}
and $\nabla_{\eta_G} \mL_G = 0$ at  $\mu=\hat{\mu}$ and $C=\hat{C}$. For the SN or CSN subclass variational approximation $q(\theta)$, let $\eta = (\mu^T, \vec(C)^T, \lambda^T)^T$ denote its parameter. From \eqref{score fn grad}, gradient of the lower bound $\mL$ of $q(\theta)$ is $\nabla_\eta \mL = \E_q \{h(\theta)  \nabla_\eta \log q (\theta)\} $. Since $q(\theta)$ reduces to $q_G(\theta)$ at $\lambda = 0$, 
\[
\nabla_\eta \mL|_{\lambda=0} = \E_{q_G} \left[ h_G(\theta)  \{\nabla_\eta \log q (\theta)\}_{\lambda = 0} \right].
\] 

If $q(\theta)$ is the SN from \eqref{SN}, then 
$\log q(\theta) = \log 2  + \log q_G(\theta) + \log \Phi\{\lambda^T (\theta - \mu) \}$, and
\[
\begin{aligned}
\nabla_\mu \log q (\theta) = \nabla_\mu \log q_G(\theta) - \frac{\phi\{\lambda^T (\theta - \mu) \} }{\Phi\{\lambda^T (\theta - \mu) \}} \lambda &\implies  
\nabla_\mu \log q (\theta)|_{\lambda = 0} =  \nabla_\mu \log q_G(\theta), \\
\nabla_{\vec(C)} \log q (\theta) = \nabla_{\vec(C)} \log q_G(\theta) &\implies \nabla_{\vec(C)} \log q (\theta)|_{\lambda = 0} = \nabla_{\vec(C)} \log q_G(\theta), \\
\nabla_\lambda \log q (\theta) = \frac{\phi\{\lambda^T (\theta - \mu) \} }{\Phi\{\lambda^T (\theta - \mu) \}} (\theta - \mu)  &\implies 
\nabla_\lambda \log q (\theta)|_{\lambda = 0}= b(\theta - \mu).
\end{aligned}
\]
Next, suppose $q(\theta)$ is the CSN from \eqref{Proposed CSN}, then
\[
\log q(\theta) = d\log(2) - \frac{d}{2}\log(2\pi) - \frac{v^T v}{2} - \log |C| + \sum_{i=1}^d \{\log \Phi (\lambda_i v_i) + \log \tau_i\},
\]
where $v = D_\tau C^{-1} (\theta - \mu) + b \delta = D_\tau z + b \delta$ , $D_\tau = \diag(\tau)$, $\delta_i = \lambda_i/\sqrt{1 + \lambda_i^2}$ and $\tau_i = \sqrt{1-b^2\delta_i^2}$. We have
\[
\frac{d\delta_i}{d\lambda_i} = \frac{1}{(1+\lambda_i^2)^{3/2}}, \quad 
\frac{d\tau_i}{d\lambda_i} = - \frac{b^2 \lambda_i}{\tau_i(1+\lambda_i^2)^2}, \quad 
\frac{d v_i}{d\lambda_i} = \frac{d\tau_i}{d\lambda_i} z_i + b\frac{d\delta_i}{d\lambda_i}.
\]
Let $\omega = \phi(\lambda \odot v) \lambda / \Phi(\lambda \odot v)$. When $\lambda = 0$, $\delta = 0$, $\tau = \bfone$, $D_\tau = I_d$, $v = z$ and $\omega = 0$. Thus
\[
\begin{gathered}
\nabla_\mu \log q (\theta) = C^{-T} D_\tau (v - \omega)
\implies  
\nabla_\mu \log q (\theta)|_{\lambda = 0} =  \nabla_\mu \log q_G(\theta),  \\
\nabla_{\vec(C)} \log q (\theta) = \vec[\{ - C^{-T} D_\tau (\omega - v)(\theta - \mu)^T - I_d \}C^{-T}] \\
\implies \nabla_{\vec(C)} \log q (\theta)|_{\lambda = 0} = \nabla_{\vec(C)} \log q_G(\theta), \\
\nabla_\lambda \log q (\theta) = \left[ - v_i \frac{d v_i}{d\lambda_i} + \frac{\phi(\lambda_i v_i)}{\Phi(\lambda_i v_i)} \left(v_i + \lambda_i \frac{d v_i}{d\lambda_i} \right) + \frac{1}{\tau_i}  \frac{d\tau_i}{d\lambda_i} \right]  
\implies 
\nabla_\lambda \log q (\theta)|_{\lambda = 0}= 0.
\end{gathered}
\]

Therefore, 
\[
\{\nabla_\eta \log q (\theta)\}_{\lambda = 0} 
= \begin{bmatrix} \nabla_\mu \log q (\theta) \\  \nabla_{\vec(C)} \log q (\theta) \\
\nabla_\lambda \log q (\theta) \end{bmatrix}_{\lambda = 0} 
= \begin{bmatrix}
\nabla_\mu \log q_G(\theta)  \\ \nabla_{\vec(C)} \log q_G(\theta) \\ \nabla_\lambda \log q (\theta)|_{\lambda = 0} 
\end{bmatrix}
= \begin{bmatrix} \nabla_{\eta_G} \log q_G (\theta)  \\ \nabla_\lambda \log q (\theta)|_{\lambda = 0} \end{bmatrix},
\]
where $\nabla_\lambda \log q (\theta)|_{\lambda = 0} = b(\theta -\mu)$ for the SN and $\nabla_\lambda \log q (\theta)|_{\lambda = 0} = 0$ for the CSN subclass. From \eqref{identity}, $\E_{q_G}\{ h_G(\theta) (\theta-\mu)\} = 0$ at $\mu=\hat{\mu}$, $C=\hat{C}$. Hence $\nabla_\eta \mL = 0$ at $\mu=\hat{\mu}$, $C=\hat{C}$ and $\lambda=0$ for both the SN and CSN subclass, and $\mL$ is stationary at this point. 
\end{proof}

\section{Proof of Lemma 1}
First, we present and prove Lemma S1.

\begin{lemma} \label{LemS1}
Let $q(\theta)$ and $\tilde{q}(\theta)$ be respectively pdfs of $\CSN_{d,q}(\mu, \Sigma, D, \nu, \Delta)$ and $\CSN_{d,q}(\mu + \Sigma s, \Sigma, D, \nu -  D \Sigma s, \Delta)$. Then
\begin{enumerate}[(i)]
\item $\exp(s^T \theta) \phi_d(\theta| \mu, \Sigma) 
= \exp(\mu^T s + s^T \Sigma s/2) \phi_d(\theta| \mu + \Sigma s, \Sigma)$,
\item $\exp(s^T \theta) q(\theta) = \exp(\mu^T s + s^T \Sigma s/2) \tilde{q}(\theta) \dfrac{\Phi_q(0| \nu - D \Sigma s, \Delta + D \Sigma D^T)}{\Phi_q(0| \nu, \Delta + D \Sigma D^T)}$.
\end{enumerate}
\end{lemma}
\begin{proof}
For (i), 
\[
\begin{aligned}
\exp(s^T \theta) \phi_d(\theta| \mu, \Sigma) 
&= (2\pi)^{-d/2} |\Sigma|^{-1/2} \exp \big[- \tfrac{1}{2}\{ \theta^T \Sigma^{-1} \theta - 2 \theta^T \Sigma^{-1} (\mu + \Sigma s) + \mu^T \Sigma^{-1} \mu \} \big]  \\
&= \phi_d(\theta |\mu + \Sigma s, \Sigma) \exp(\mu^T s + s^T \Sigma s/2).
\end{aligned}
\]
For (ii), using the result in (i) we have 
\[
\begin{aligned}
\exp(s^T \theta) q(\theta) &=  \exp(\mu^T s + s^T \Sigma s/2) \phi_d(\theta |\mu + \Sigma s, \Sigma) \frac{\Phi_q (D(\theta - \mu)|\nu, \Delta)}{\Phi_q (0| \nu, \Delta + D \Sigma D^T)} \\
&= \exp(\mu^T s + s^T \Sigma s/2) \tilde{q}(\theta) \frac{\Phi_q (0| \nu - D \Sigma s, \Delta + D \Sigma D^T)}{ \Phi_q (D(\theta - \mu - \Sigma s) | \nu-D \Sigma s, \Delta)} \frac{ \Phi_q (D(\theta - \mu)|\nu, \Delta)}{\Phi_q (0| \nu, \Delta + D \Sigma D^T)}  \\
& =  \exp(\mu^T s + s^T \Sigma s/2) \tilde{q}(\theta) \dfrac{\Phi_q(0| \nu - D \Sigma s, \Delta + D \Sigma D^T)}{\Phi_q(0| \nu, \Delta + D \Sigma D^T)}.
\end{aligned}
\]
\end{proof}

The result in Lemma 1 follows directly from Lemma \ref{LemS1}(ii) by replacing $\mu$ by $\mu^*$, $\Sigma$ by $\Sigma^*$, $D$ by $D^*$ and setting $\nu =0$ and $\Delta = I_d$. Note that $\Phi_d(0|- D_\lambda D_\tau^{-1} C^Ts, I_d + D_\lambda^2) = \Phi( D_\alpha C^Ts) $. The moment generating function of $\tilde{q}(\theta)$ is
\[
M(t) = \frac{\Phi_d(D_\alpha C^T (t+s))}{\Phi(D_\alpha C^T s)} \exp \{t^T (\mu^* + \Sigma^*s) + t^T\Sigma^* t/2\},
\]
and the log cumulant function is
\[
K(t) = \log M(t) = \sum_{j=1}^d \log \Phi(\alpha_j C_{\cdot j}^T (t+s)) + t^T(\mu^* + \Sigma^*s)  + \frac{t^T \Sigma^* t}{2} + \text{constant}.
\]
Substituting $t=0$ in the expressions below yields the mean and covariance of $\tilde{q}(\theta)$:
\[
\begin{aligned}
\nabla_t K(t) &= \sum_{j=1}^d \zeta_1(\alpha_j C_{\cdot j}^T (t+s)) \alpha_j C_{\cdot j}  + \mu^* + \Sigma^*s +\Sigma^* t, \\
\nabla_t^2 K(t)&= \sum_{j=1}^d \zeta_2(\alpha_j C_{\cdot j}^T (t+s)) \alpha_j^2 C_{\cdot j} C_{\cdot j}^T + \Sigma^*.
\end{aligned}
\]

\section{Normal sample}  
The induced prior for $\theta_2$ is $p(\theta_2) = {b_0}^{a_0} \exp\{- a_0 \theta_2 - b_0 \exp(-\theta_2)\}/\Gamma(a_0)$ and 
\[
\begin{aligned}
\log p(y,\theta) = c_* - \left( a_0 + n/2 \right) \theta_2 - \exp(- \theta_2) \left(b_0 + \sum\nolimits_{i=1}^n (y_i - \theta_1)^2/2 \right) - \theta_1^2/(2\sigma_0^2),
\end{aligned}
\]
where $c_* = a_0 \log b_0 - \log \Gamma(a_0) - \frac{1}{2} \log(\sigma_0^2) - \frac{n+1}{2}\log(2\pi)$. From Lemma 1, taking $s$ as $-e_2$,
\[
\exp(-e_2^T \theta) q(\theta) = 2^d \Phi_d(- D_\alpha C^T e_2) \exp(-e_2^T \mu^* + e_2^T \Sigma^* e_2/2) \tilde{q}(\theta) = M \tilde{q}(\theta),
\]
where $M$, $\tilde{\mu}$ and $\tilde{\Sigma}$ are defined as in Lemma 1. Using this result and taking expectations,
\[
\begin{aligned}
\E_q \{ \log p(y,\theta) \} &= c_* - \left( \frac{n }{2} + a_0 \right) \E_q(\theta_2) - \int \underbrace{\exp(-\theta_2) q(\theta)}_{ M\tilde{q}(\theta)} \left(b_0 + \frac{\sum_{i=1}^n (y_i - \theta_1)^2}{2}  \right) d\theta- \frac{\E_q(\theta_1^2)}{2\sigma_0^2} \\
&= c_* - \left( \frac{n }{2} + a_0 \right) \mu_2  -M \left(b_0 + \frac{n \tilde{\Sigma}_{11} + \sum_{i=1}^n (y_i - \tilde{\mu}_1)^2 }{2}  \right) - \frac{e_1^T(CC^T + \mu\mu^T) e_1}{2\sigma_0^2}. \\
&= c_* - (a_0 + n/2) \mu_2 - M (b_0 + T/2) - (\Sigma_{11} + \mu_1^2 )/(2\sigma_0^2) ,
\end{aligned}
\]
where $T = \sum_{i=1}^n (y_i - \tilde{\mu}_1)^2 + n \tilde{\Sigma}_{11}$.

In the univariate case, 
\[
\E_q \{ \log p(y,\theta) \}= a_0 \log b_0 - \log \Gamma(a_0) - \frac{n}{2}\log(2\pi) - (a_0 + n/2) \mu - M T,
\]
where $M = 2\Phi(-\alpha \sigma) \exp(\sigma^2 \tau^2/2 - \mu + b\sigma \alpha)$, $T = b_0 + \sum_{i=1}^n y_i^2/2$. The lower bound can be written as a function of $(\sigma, \lambda)$ only,
\[
\mL(\sigma, \lambda) = a_0 \log b_0 - \log \Gamma(a_0) - \frac{n}{2}\log(2\pi) - (a_0 + n/2) \{ f(\sigma, \lambda)  +1 \} + H_q.
\]
since $\mL$ is maximized at $\mu = f(\sigma, \lambda) = \sigma^2 \tau^2/2 + b\sigma \alpha - \log(a_0 + n/2) + \log\{2T\Phi(-\sigma \alpha)\}$.

To find the normalizing constant of the posterior density in the bivariate case, we first integrate $p(y, \theta)$ with respect to $\theta_1$. The integral with respect to only $\theta_2$ can then be evaluated numerically.
\[
\begin{aligned}
p(y) &= \int p(y, \theta) d\theta \\
&=  \int \frac{b_0^{a_0}}{\Gamma(a_0)} \e^{-a_0 \theta_2 - b_0 \exp(-\theta_2)} \cdot \frac{1}{\sqrt{2\pi \sigma_0^2}} \e^{-\theta_1^2/(2\sigma_0^2)} \prod_{i=1}^n \frac{1}{\sqrt{2\pi} \e^{\theta_2/2}} \e^{-(y_i-\theta_1)^2/(2\e^{\theta_2})} d\theta_1 d\theta_2 \\
&=  \int \frac{b_0^{a_0}\exp\{-(a_0 + n/2) \theta_2 - (b_0 + \sum_{i=1}^n y_i^2/2) \exp(-\theta_2)\} }{\Gamma(a_0)(2\pi)^{(n+1)/2}\sqrt{\sigma_0^2}} \int \e^{ -\frac{(\theta_1 - m)^2}{2v} + \frac{m^2}{2v} }  d\theta_1 d\theta_2 \\
&=  \int \frac{b_0^{a_0}\exp\{m^2/(2v) - (a_0 + n/2) \theta_2 - (b_0 + \sum_{i=1}^n y_i^2/2) \exp(-\theta_2)\} \sqrt{v}}{\Gamma(a_0)(2\pi)^{n/2}\sqrt{\sigma_0^2}} d\theta_2,
\end{aligned}
\]
where  $v = \{n\exp(-\theta_2) +\sigma_0^{-2}\}^{-1}$, $m = v n \bar{y} \exp(-\theta_2)$ and $\bar{y} = \sum_{i=1}^n y_i/n$.

\section{Poisson generalized linear model}
Consider a dataset \citep{Scotto1974} on the incidence of non-melanoma skin cancer by age group (15--24, \dots, 75--84, 84+) in Dallas-Fort Worth and Minneapolis-St. Paul, given in Table 6 of \cite{Gart1979}. Let $y_i \sim \text{Poisson}(\gamma_i)$ be the number of skin cancers in group $i$ out of a population $T_i$, where the expected rate is modeled as  $\log (\mu_i/T_i) = x_i^T \theta$ for $i=1, \dots, n$ and $\log(T_i)$ acts as an offset. The covariates $x_i$ are dummy variables for age group and town, such that $n=16$ and $d=9$. A normal prior $\N(0, \sigma_o^2I_d)$ is placed on $\theta$, where $\sigma_0 = 100$. Applying Lemma 1 again, we obtain
\[
\E_q \{\log p(y, \theta) \}= y^T X\mu - \frac{\tr(\Sigma) + \mu^T \mu }{2\sigma_0^2} - 2^d \sum_{i=1}^n T_i\exp\{ \mu^{*T}  x_i + x_i^T \Sigma^* x_i /2 \} \Phi(D_\alpha C^T x_i) + c_*,
\]
where $c_* =  - \frac{d}{2} \log(2\pi \sigma_o^2) + \sum_{i=1}^n \{y_i \log(T_i) - \log(y_i!) \}$. Optimal parameters for the variational approximations are obtained using BFGS via {\tt Optim} in {\tt R}. 
\begin{table}[tb!]
\centering
\begin{small}
\begin{tabular}{@{}lcccccccccc@{}}
\hline
 & $\mL$ & $\theta_1$  & $\theta_2$  & $\theta_3$  & $\theta_4$  & $\theta_5$  & $\theta_6$  & $\theta_7$  & $\theta_8$  & $\theta_9$ \\ \hline
Gaussian & -115.027 & 94.9 & 95.5 & 95.2 & 94.9 & 94.9 & 94.9 & 95.0 & 95.4 & 99.3 \\ 
  SN & -115.011 & 98.7 & 98.6 & 98.6 & 98.7 & 98.6 & 98.6 & 98.7 & 98.6 & 99.3 \\ 
  CSNC & -115.009 & 98.6 & 98.9 & 98.7 & 98.7 & 98.6 & 98.6 & 98.7 & 98.7 & 99.4 \\ 
  CSNLU & -115.008 & 98.6 & 98.8 & 98.6 & 98.7 & 98.5 & 98.5 & 98.6 & 98.7 & 99.3 \\ 
\hline
\end{tabular}
\caption{Lower bounds and accuracies (\%) of variational approximations.}\label{T2}
\end{small}
\end{table}
\begin{figure}[tb!]
\centering
\includegraphics[width=0.9\textwidth]{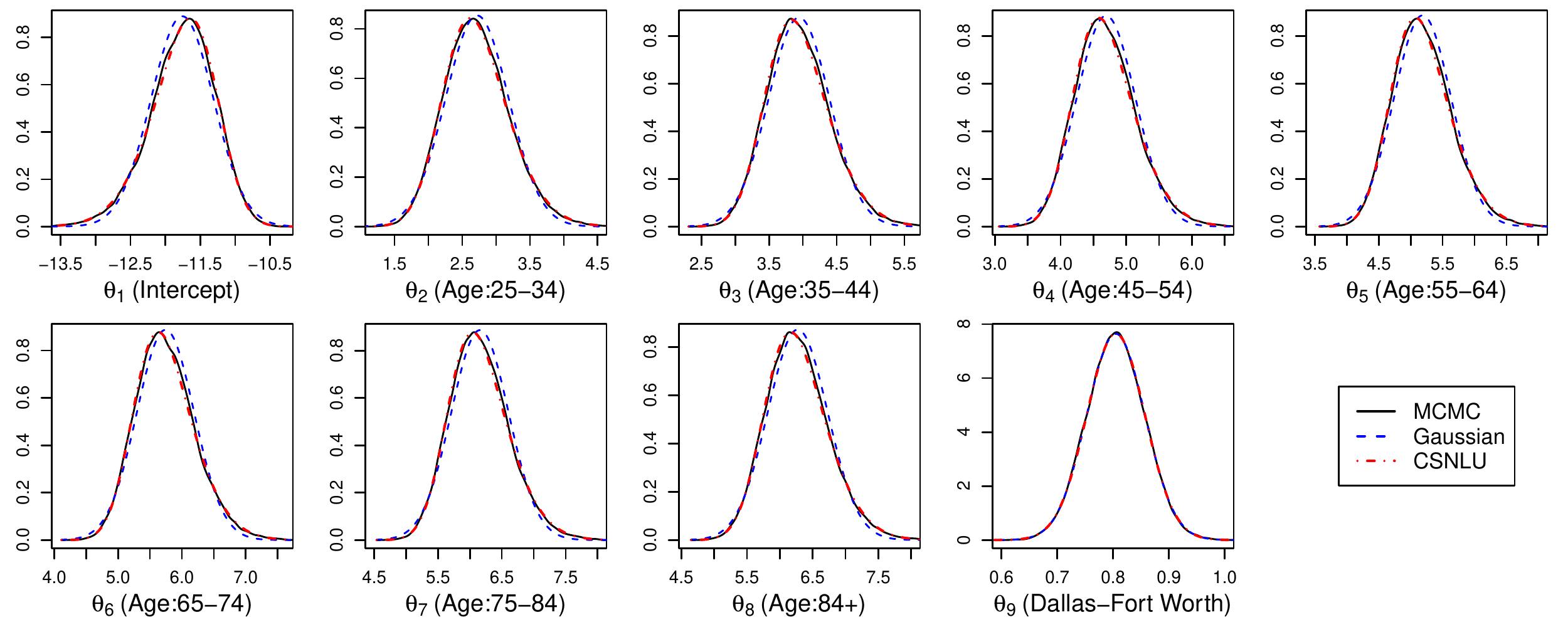}
\caption{Marginal posterior densities estimated using MCMC, Gaussian, and CSNLU. }\label{F7}
\end{figure}
Table \ref{T2} shows that the Gaussian variational approximation has an accuracy of about 95\% compared to MCMC for the marginal densities of all parameters except $\theta_9$, while the approximations incorporating skewness can attain an accuracy of 99\%. This improvement in approximation of the marginal densities can also be observed in Figure \ref{F7}. While CSNC and CSNLU attain a slightly higher lower bound than the SN, indicating a mildly better approximation of the joint posterior, the differences are hardly distinguishable marginally. Examination of the bivariate MCMC draws do not indicate a need for a bounding line in each dimension, which explains why SN performs almost as well as the CSN subclass.

\section{Reparametrization trick gradients} 
 If $C$ is the Cholesky factor, then 
\[
\begin{aligned}
\nabla_\eta \theta = \begin{bmatrix}  \nabla_\mu \theta \\  \nabla_\lambda \theta \\  \nabla_{\vech(C)} \theta  \end{bmatrix}
= \begin{bmatrix}
I_d \\
\diag\{ \widetilde{w}_1 - (1-b^2) D_\lambda w_2 \}D_\kappa^3C^T \\
\El(z \otimes I)
\end{bmatrix},
\end{aligned}
\]
On the other hand, if $C  = L U$,  then
\[
\begin{aligned}
\nabla_\eta \theta = \begin{bmatrix}  \nabla_\mu \theta \\  \nabla_\lambda \theta \\  \nabla_{\vech(L)} \theta \\
\nabla_{\vech_u(U)} \theta  \end{bmatrix}
= \begin{bmatrix}
I_d \\
\diag\{ \widetilde{w}_1 - (1-b^2) D_\lambda w_2 \}D_\kappa^3C^T \\
\El(Uz \otimes I) \\ 
\Eu (z \otimes L^T)
\end{bmatrix}.
\end{aligned}
\]

\section{Natural gradients}
Let $\El$, $\Eu$ and $\Ed$ denote elimination matrices such that $\El \vec(X) = \vech(X)$, $\Eu \vec(X) = \vech_u (X)$ and $\Ed \vec(X) = \diag(X)$. Conversely, $\El^T \vech(X) = \vec(X_\ell)$, $\Eu^T \vech_u(X) = \vec(X_u)$ and $\Ed^T \diag(X) = \vec\{ \dg(X) \}$, where $\dg(X)$ is a diagonal matrix derived from $X$ by setting all non-diagonal elements to zero. Let $\K$ be the commutation matrix such that $\K \vec(X) = \vec(X^T)$.  We use $\otimes$ to denote the Kronecker product. A good reference for vector differential calculus is \cite{Magnus1999}. 

Let $\ell = \log q(\theta, w)$. For taking expectations, note that $\E(z) = 0$, $\E(zz^T) = I_d$, $\E(\widetilde{w}) = 0$, $\E(\widetilde{w} \widetilde{w}^T) = (1-b^2) I_d$, $\E(z|w) = D_\alpha \widetilde{w}$, $\E(z \widetilde{w}^T) = (1-b^2) D_\alpha$ and 
\[
\E[\nabla_\mu \ell z^T] = C^{-T} \{D_\kappa^{-2} - (1-b^2) D_\kappa^{-1}D_\lambda D_\alpha \} =C^{-T}.
\]

\subsection{Cholesky factorization (Proof of Theorem 1)} \label{AppendixE1}
If $C$ is the Cholesky factor, the first order derivatives are
\[
\begin{aligned}
\nabla_\mu \ell &= C^{-T} D_\kappa^{-1} \{D_\kappa^{-1} z - D_\lambda \widetilde{w} \} , \\
\nabla_\lambda \ell &= (1-b^2) (\alpha \odot \kappa) - D_\lambda \{(1-b^2)  z^2 + \widetilde{w}^2\} + \diag(2/\kappa - \kappa) (z \odot\widetilde{w}), \\
\nabla_{\vech(C)} \ell & = \vech \{(\nabla_\mu \ell)z^T - C^{-T} \}.
\end{aligned}
\]
The second order derivatives are
\[
\begin{aligned}
\nabla_\mu^2 \ell &= - C^{-T} D_{\kappa}^{-2} C^{-1},
\\
\nabla_{\mu, \lambda }^2 \ell &= C^{-T} \diag\{ 2(1-b^2)(\lambda \odot z)
 - (2/\kappa - \kappa) \odot \widetilde{w} \},
\\
\nabla_{ \vech(C), \mu}^2 \ell &= \El\{ (z \otimes \nabla_\mu^2 \ell) - (C^{-1} \otimes  \nabla_\mu \ell) \},
\\
\nabla_\lambda^2 \ell &= (1-b^2) \diag \{ 2\kappa^4 - \kappa^2 + (2\alpha + \alpha\kappa^2) \odot z \odot \widetilde{w} - z^2 \}  - \diag(\widetilde{w}^2),
\\
\nabla_{\vech(C)}^2 \ell &= \El[  (zz^T \otimes \nabla_\mu^2 \ell ) + \{(C^{-1} \otimes C^{-T}) - (C^{-1} \otimes (\nabla_\mu \ell) z^T) - (z ( \nabla_\mu \ell)^T \otimes C^{-T} ) \}\K] \El^T,
\\
\nabla_{\vech(C), \lambda}^2 \ell &= \El(z \otimes \nabla_{\mu, \lambda }^2 \ell) .
\end{aligned}
\]
Taking expectations, $\E[\nabla_{\mu, \lambda }^2 \ell] = 0$, $\E[\nabla_{\mu, \vech(C)}^2 \ell] = 0$,
\[
\begin{aligned}
\E[\nabla_\lambda^2 \ell] &=- (1-b^2) \diag( 2\kappa^2 - \kappa^4), \\
\E[\nabla_{\vech(C)}^2 \ell] &= - \El \{(C^{-1} \otimes C^{-T}) \K + (I_d \otimes C^{-T} D_{\kappa}^{-2} C^{-1} )\} \El^T, \\
\E[\nabla_{\vech(C), \lambda}^2 \ell] &= (1-b^2) \El(I \otimes C^{-T})  \Ed^T \diag(\alpha \kappa).
\end{aligned}
\]
Note that
\[
\begin{aligned}
\mathcal{I}_{33} &=  \El \{\K (C^{-T} \otimes C^{-1}) + (I_d \otimes C^{-T} D_{\kappa}^{-2} C^{-1} )\} \El^T \\
&= \El \{\K (C^{-T} \otimes I_d)(I_d \otimes C^{-1}) + (I_d \otimes C^{-T}) (I_d \otimes D_{\kappa}^{-2}) (I_d \otimes C^{-1} )\} \El^T \\
&= \El  (I_d \otimes C^{-T} ) (\K + I_d \otimes D_{\kappa}^{-2}) (I_d \otimes C^{-1} )\El^T.
\end{aligned}
\]
Thus we obtain the Fisher information matrix of $q(\theta, w)$, which is given by 
\[
\mathcal{I}_{\theta, w}(\eta) = \begin{bmatrix}
(C D_\kappa^2 C^T)^{-1} & 0 & 0 \\
0 & \mathcal{I}_{22} &  \mathcal{I}_{32}^T \\
0 & \mathcal{I}_{32} & \mathcal{I}_{33}
\end{bmatrix}, \quad 
\begin{aligned}
\mathcal{I}_{22} &= (1-b^2) \diag( 2\kappa^2 - \kappa^4), \\
\mathcal{I}_{32} &= - (1-b^2) \El(I_d \otimes C^{-T}) \Ed^T \diag(\alpha \kappa), \\
\mathcal{I}_{33} &=\El  (I_d \otimes C^{-T} ) (\K + I_d \otimes D_{\kappa}^{-2}) (I_d \otimes C^{-1} )\El^T.
\end{aligned}
\]
The natural gradient is given by $\mathcal{I}_{\theta, w}(\eta)^{-1}\nabla_\eta \mL$, where $\mathcal{I}_{\theta, w}(\eta)^{-1}$ is
\[
\begin{aligned}
\begin{bmatrix}
C D_{\kappa}^2 C^T & 0 & 0 \\
0 & \mathcal{I}_{22}^{-1}  + \mathcal{I}_{22}^{-1}  \mathcal{I}_{23}(\mathcal{I}_{33} - \mathcal{I}_{32} \mathcal{I}_{22}^{-1} \mathcal{I}_{23})^{-1} \mathcal{I}_{32} \mathcal{I}_{22}^{-1} &  - \mathcal{I}_{22}^{-1}  \mathcal{I}_{23}(\mathcal{I}_{33} - \mathcal{I}_{32} \mathcal{I}_{22}^{-1} \mathcal{I}_{23})^{-1}   \\
0 &  - (\mathcal{I}_{33} - \mathcal{I}_{32} \mathcal{I}_{22}^{-1} \mathcal{I}_{23})^{-1} \mathcal{I}_{32} \mathcal{I}_{22}^{-1} &  (\mathcal{I}_{33} - \mathcal{I}_{32} \mathcal{I}_{22}^{-1} \mathcal{I}_{23})^{-1} & \\
\end{bmatrix}.
\end{aligned}
\]
We have 
\[
\mathcal{I}_{32} \mathcal{I}_{22}^{-1} \mathcal{I}_{23} 
= \El(I_d \otimes C^{-T}) \Ed^T \diag ( \tfrac{1-\kappa^2}{2-\kappa^2} ) \Ed (I_d \otimes C^{-1}) \El^T
\]
and $\mathcal{I}_{33} - \mathcal{I}_{32} \mathcal{I}_{22}^{-1} \mathcal{I}_{23} =  \El  (I_d \otimes C^{-T} )  N_\lambda (I_d \otimes C^{-1}) \El^T$, where $N_\lambda = \K + I \otimes D_\kappa^{-2} - \Ed^T \diag ( \tfrac{1-\kappa^2}{2-\kappa^2} ) \Ed$. Note that 
\[
\begin{aligned}
\El N_\lambda \El^T 
&= \El \K \El^T + \El (I \otimes D_\kappa^{-2})\El^T - \El \Ed^T \diag ( \tfrac{1-\kappa^2}{2-\kappa^2} ) \Ed \El^T \\
&= \diag \left[ \vech \left\{ I_d + \kappa^{-2} \bfone^T - \diag(\tfrac{1-\kappa^2}{2-\kappa^2}) 
\right\}  \right]\\
&= \diag \left[ \vech \left\{ \kappa^{-2} \bfone^T + \diag(\tfrac{1}{2-\kappa^2}) \right\}  \right],
\end{aligned}
\]
and $ (\El N_\lambda \El^T)^{-1} =\diag [ \vech \{ K - \diag(\tfrac{\kappa^4}{2}) \} ]$. We claim that
\[
(\mathcal{I}_{33} - \mathcal{I}_{32} \mathcal{I}_{22}^{-1} \mathcal{I}_{23})^{-1} = \El(I_d \otimes C) \El^T (\El N_\lambda \El^T)^{-1} \El(I_d \otimes C^T) \El^T,
\]
since
\[
\begin{aligned}
&\{ \El  (I_d \otimes C^{-T} ) N_\lambda  (I_d \otimes C^{-1} )\El^T \} \{\El(I_d \otimes C) \El^T (\El N_\lambda \El^T)^{-1} \El(I_d \otimes C^T) \El^T\} \\
&= \El  (I_d \otimes C^{-T} ) \El^T \El N_\lambda  (I_d \otimes C^{-1} )(I_d \otimes C) \El^T (\El N_\lambda \El^T)^{-1} \El(I_d \otimes C^T) \El^T \\
&= \El  (I_d \otimes C^{-T} )  \El^T \El(I_d \otimes C^T) \El^T \\
&= \El  (I_d \otimes C^{-T} ) (I_d \otimes C^T) \El^T = I_{d(d+1)/2}.
\end{aligned}
\]
In the 2nd and 4th lines, we use the result , $\El^T \El (P^T \otimes Q) \El^T = (P^T \otimes Q) \El^T$,
(and its transpose) for any lower triangular $d \times d$ matrices $P$ and $Q$, from Lemma 4.2 (i) of \cite{Magnus1980}. Next
\[
\begin{aligned}
 (\mathcal{I}_{33} - \mathcal{I}_{32} \mathcal{I}_{22}^{-1} \mathcal{I}_{23})^{-1} \nabla_{\vech(C)} \mL 
& = \El(I_d \otimes C) \El^T (\El N_\lambda \El^T)^{-1} \El(I_d \otimes C^T) \El^T \nabla_{\vech(C)} \mL\\
& = \El(I_d \otimes C) \El^T \diag [ \vech \{ K - \diag(\tfrac{\kappa^4}{2}) \} ] \vech(G) \\
& = \El(I_d \otimes C) \El^T \vech[ \{ K - \diag(\tfrac{\kappa^4}{2}) \} \odot G ] \\
& = \El(I_d \otimes C) \vec[ \{ K - \diag(\tfrac{\kappa^4}{2}) \} \odot G ] \\
& = \vech( C [G \odot  \{ K - \diag(\tfrac{\kappa^4}{2}) \}] ),
\end{aligned}
\]
and
\[
\begin{aligned}
& -(\mathcal{I}_{33} - \mathcal{I}_{32} \mathcal{I}_{22}^{-1} \mathcal{I}_{23})^{-1} \mathcal{I}_{32} \mathcal{I}_{22}^{-1} \nabla_\lambda \mL \\
& = \{ \El(I_d \otimes C) \El^T (\El N_\lambda \El^T)^{-1} \El(I_d \otimes C^T) \El^T \} \El(I_d \otimes C^{-T}) \Ed^T \diag( \lambda /(2 - \kappa^2)) \nabla_\lambda \mL \\
& = \El(I_d \otimes C) \El^T (\El N_\lambda \El^T)^{-1} \El (I_d \otimes C^T) (I_d \otimes C^{-T}) \vec\{\diag( \lambda/(2 - \kappa^2) \odot \nabla_\lambda \mL)\} \\
& = \El(I_d \otimes C) \El^T \vech \{ K - \diag(\tfrac{\kappa^4}{2}) \}\odot \vech\{\diag( \lambda/(2 - \kappa^2) \odot \nabla_\lambda \mL)\} \\
& = \El(I_d \otimes C) \El^T \vech\{\diag( \alpha \kappa/2 \odot \nabla_\lambda \mL)\} \\
& = \vech\{C \diag(\tfrac{\alpha\kappa}{2} \odot \nabla_\lambda \mL) \}.
\end{aligned}
\]
Thus the natural gradient is
\[
\begin{aligned}
\widetilde{\nabla}_\mu \mL &= C D_\kappa^2 C^T \nabla_ \mu \mL \\
\widetilde{\nabla}_{\vech(C)} \mL &=  \vech( C A_1) \\
\widetilde{\nabla}_\lambda \mL &= \mathcal{I}_{22}^{-1} \nabla_\lambda \mL -  \mathcal{I}_{22}^{-1} \mathcal{I}_{23} \widetilde{\nabla}_{\vech(C)} \mL \\
&= \tfrac{1}{(1-b^2) (2\kappa^2 - \kappa^4)} \odot \nabla_\lambda \mL +  \diag \left( \tfrac{\lambda}{ 2 - \kappa^2 } \right) \Ed (I_d \otimes C^{-1} ) \El^T \vech( C A_1)  \\
&= \tfrac{1}{(1-b^2) (2\kappa^2 - \kappa^4)} \odot \nabla_\lambda \mL + \diag \left( \tfrac{\lambda}{ 2 - \kappa^2 } \right) \Ed \vec( A_1) \\
&= \tfrac{1}{(1-b^2) (2\kappa^2 - \kappa^4)} \odot \nabla_\lambda \mL + \tfrac{\lambda}{2 - \kappa^2} \odot \diag( A_1).
\end{aligned}
\]

\subsection{LU decomposition (Proof of Theorem 2)}
For the first order derivatives, $\nabla_\mu \ell$ and $\nabla_\lambda \ell$ remain the same as before while
\[
\begin{aligned}
\nabla_{\vech(L)} \ell & = \vech \{(\nabla_\mu \ell)z^T U^T - L^{-T} \}, \quad
\nabla_{\vech_u(U)} \ell  = \vech_u \{ L^T(\nabla_\mu \ell)z^T \}.
\end{aligned}
\]
For the second order derivatives, $\nabla_\mu^2 \ell$, $\nabla_{\mu, \lambda }^2 \ell$ and $\nabla_\lambda^2 \ell$ remain unchanged, while
\[
\begin{aligned}
\nabla_{\vech(L), \mu}^2 \ell &= \El \{ (Uz \otimes \nabla_\mu^2 \ell) - (L^{-1} \otimes  \nabla_\mu \ell)\} , \\
\nabla_{\vech_u(U), \mu}^2 \ell &= \Eu \{(z \otimes L^T  \nabla_\mu^2 \ell) - (C^{-1} \otimes L^T \nabla_\mu \ell) \}, \\
\nabla_{\vech(L), \lambda}^2 \ell &= \El (Uz \otimes \nabla_{\mu, \lambda }^2 \ell ), \\
\nabla_{\vech_u(U), \lambda}^2 \ell &= \Eu [ z \otimes L^T \nabla_{\mu, \lambda }^2 \ell ], \\
\nabla_{\vech(L), \vech_u(U)}^2 \ell &= \El [  Uzz^T \otimes (\nabla_\mu^2 \ell) L - \{ Uz ( \nabla_\mu \ell)^TL  \otimes C^{-T} \} \K ] \Eu^T, \\
\nabla_{\vech(L)}^2 \ell &= \El [ \{ (L^{-1}- Uz  ( \nabla_\mu \ell)^T) \otimes L^{-T} - L^{-1} \otimes (\nabla_\mu \ell) z^T U^T \} \K + Uzz^T U^T \otimes \nabla_\mu^2 \ell ] \El^T, \\
\nabla_{\vech_u(U)}^2 \ell  &= - \Eu \{ (z  ( \nabla_\mu \ell)^T L \otimes U^{-T} + U^{-1} \otimes L^T  ( \nabla_\mu \ell) z^T) \K + (zz^T \otimes U^{-T} D_\kappa^{-2} U^{-1}) \} \Eu^T.
\end{aligned}
\]
Taking expectations, $\E[\nabla_{\vech(L), \mu}^2 \ell] = 0$, $\E[\nabla_{\vech_u(U), \mu}^2 \ell] =0$,
\[
\begin{aligned}
\E[\nabla_{\vech_u(U), \lambda}^2 \ell] &= (1-b^2) \Eu(I_d \otimes U^{-T}) \Ed^T \diag(\alpha \kappa) = 0, \\
\E[\nabla_{\vech(L), \lambda}^2 \ell] &= (1-b^2) \El(U \otimes C^{-T}) \Ed^T \diag(\alpha \kappa) . \\
\E[\nabla_{\vech(L), \vech_u(U)}^2 \ell] 
&= \El \{U \otimes (\nabla_\mu^2 \ell) L -  (I_d  \otimes C^{-T}) \K\}\Eu^T \\
&= - \El \{U \otimes C^{-T} D_\kappa^{-2} U^{-1}+  (UU^{-1}  \otimes C^{-T}) \K\} \Eu^T \\
&= - \El (U \otimes C^{-T}) (\K + I_d\otimes D_\kappa^{-2}) (I_d \otimes U^{-1})\Eu^T , \\
\end{aligned}
\]

\[
\begin{aligned}
\E[\nabla_{\vech(L)}^2 \ell] 
&=  - \El \{ (L^{-1} \otimes L^{-T})  \K - (UU^T \otimes \nabla_\mu^2 \ell) \} \El^T \\
&= - \El \{ (UC^{-1} \otimes C^{-T}U^T)  \K + (UU^T \otimes C^{-T} D_\kappa^{-2} C^{-1}) \} \El^T \\
&= - \El  (U \otimes C^{-T}) (\K+ I_d \otimes D_\kappa^{-2} ) (U^T \otimes C^{-1} ) \El^T, \\
\E[\nabla_{\vech_u(U)}^2 \ell]  &= - \Eu \{ 2(U^{-1} \otimes U^{-T})\K + (I_d \otimes U^{-T} D_\kappa^{-2} U^{-1}) \} \Eu^T \\
&= - \Eu  (I_d \otimes U^{-T}) (I_d \otimes D_\kappa^{-2}) (I_d \otimes U^{-1})\Eu^T.
\end{aligned}
\]
Note that $\Eu (U^{-1} \otimes U^{-T})\K \Eu^T = 0$ \citep{Tan2024b}. Thus we obtain the Fisher information matrix of $q(\theta, w)$, which is given by 
\[
\mathcal{I}_{\theta, w}(\eta) = \begin{bmatrix}
(C D_\kappa^2 C^T)^{-1} & 0 & 0 \\
0 & \mathcal{I}_{22} & \mathcal{I}_{32}^T \\
0 & \mathcal{I}_{32} & \mathcal{I}_{33}
\end{bmatrix}, \quad 
\begin{aligned}
\mathcal{I}_{22} &= (1-b^2) \diag( 2\kappa^2 - \kappa^4), \\
\mathcal{I}_{32} &= - (1-b^2) \begin{bmatrix}  \El(U \otimes C^{-T})  \Ed^T \diag(\alpha \kappa)  \\ 0  \end{bmatrix} 
\end{aligned} 
\]
and $\mathcal{I}_{33}$ is given by 
\[
 \begin{bmatrix}
\El  (U \otimes C^{-T}) (\K+ I_d \otimes D_\kappa^{-2} ) (U^T \otimes C^{-1} ) \El^T & \El (U \otimes C^{-T}) (\K + I_d\otimes D_\kappa^{-2}) (I_d \otimes U^{-1})\Eu^T \\
\Eu(I_d \otimes U^{-T})(\K + I_d\otimes D_\kappa^{-2})  (U^T \otimes C^{-1}) \El ^T & \Eu  (I_d \otimes U^{-T}) (I_d \otimes D_\kappa^{-2}) (I_d \otimes U^{-1})\Eu^T
\end{bmatrix}.
\]

The natural gradient is given by 
\begin{equation} \label{LU nat grad supp}
\begin{aligned}
\widetilde{\nabla}_\mu \mL &= C D_\kappa^2 C^T \nabla_\mu \mL, \\
\begin{bmatrix}
\widetilde{\nabla}_{\vech(L)} \mL \\ \widetilde{\nabla}_{\vech_u(U)} \mL 
\end{bmatrix}
& = (\mathcal{I}_{33} - \mathcal{I}_{32} \mathcal{I}_{22}^{-1} \mathcal{I}_{23})^{-1} \left\{
\begin{bmatrix} \nabla_{\vech(L)} \mL \\ \nabla_{\vech_u(U)} \mL \end{bmatrix} - \mathcal{I}_{32} \mathcal{I}_{22}^{-1} \nabla_\lambda \mL  \right\}, \\
\widetilde{\nabla}_\lambda \mL  &= \mathcal{I}_{22}^{-1} \nabla_\lambda \mL - \mathcal{I}_{22}^{-1} \mathcal{I}_{23} \begin{bmatrix}
\widetilde{\nabla}_{\vech(L)} \mL \\ \widetilde{\nabla}_{\vech_u(U)} \mL 
\end{bmatrix}.
\end{aligned}
\end{equation}

First we find $(\mathcal{I}_{33} - \mathcal{I}_{32} \mathcal{I}_{22}^{-1} \mathcal{I}_{23})^{-1} $, where
\[
\begin{gathered}
\mathcal{I}_{32} \mathcal{I}_{22}^{-1} \mathcal{I}_{23}
= \begin{bmatrix}  \El(U \otimes C^{-T})  \\ 0 \end{bmatrix} \Ed^T \diag(\tfrac{1-\kappa^2}{2-\kappa^2}) \Ed \begin{bmatrix}  \El(U \otimes C^{-T})  \\ 0 \end{bmatrix}^T, \\
\mathcal{I}_{33} - \mathcal{I}_{32} \mathcal{I}_{22}^{-1} \mathcal{I}_{23} = 
 \begin{bmatrix}
S_{11} & S_{12} \\ S_{12}^T & S_{22} 
\end{bmatrix}, \quad 
\begin{aligned}
S_{11} &= \El  (U \otimes C^{-T}) N_\lambda (U^T \otimes C^{-1} ) \El^T \\
S_{12} &= \El (U \otimes C^{-T}) (\K + I_d\otimes D_\kappa^{-2}) (I_d \otimes U^{-1})\Eu^T \\
S_{22} &=  \Eu  (I_d \otimes U^{-T}) (I_d \otimes D_\kappa^{-2}) (I_d \otimes U^{-1})\Eu^T,
\end{aligned}
\end{gathered}
\]
and $N_\lambda = \K + I_d \otimes D_\kappa^{-2} - \Ed^T \diag(\tfrac{1-\kappa^2}{2-\kappa^2}) \Ed$ are as defined previously.
Note that $\Eu (I_d \otimes D_\kappa^{-2}) \Eu ^T = \diag\{ \vech_u(\kappa^{-2} \bfone^T)\}$. Thus 
\[
\{ \Eu (I_d \otimes D_\kappa^{-2}) \Eu ^T \}^{-1} =  \diag\{ \vech_u(K)\}.
\]
We have
\[
S_{22}^{-1} = \Eu  (I_d \otimes U) \Eu^T \diag\{ \vech_u(K)\} \Eu (I_d \otimes U^T)\Eu^T,
\]
since
\[
\begin{aligned}
&\{ \underbrace{\Eu  (I_d \otimes U^{-T})}_{\Eu  (I_d \otimes U^{-T})\Eu ^T\Eu  } (I_d \otimes D_\kappa^{-2}) (I_d \otimes U^{-1}) \underbrace{\Eu^T\}  [\Eu  (I_d \otimes U) \Eu^T}_{(I_d \otimes U) \Eu^T} \diag\{ \vech_u(K)\} \Eu (I_d \otimes U^T)\Eu^T]  \\
&= \Eu  (I_d \otimes U^{-T})  \Eu ^T \{\Eu (I_d \otimes D_\kappa^{-2}) \Eu^T\}  \diag\{ \vech_u(K)\} \Eu (I_d \otimes U^T)\Eu^T \\
&= \Eu  (I_d \otimes U^{-T})  \Eu ^T \Eu (I_d \otimes U^T)\Eu^T \\
&= \Eu  (I_d \otimes U^{-T}) (I_d \otimes U^T)\Eu^T \\
&=  \Eu \Eu^T = I_{d(d-1)/2}. 
\end{aligned}
\]
In the first and 4th lines, we have made use of the property, $\Eu^T \Eu (P \otimes Q^T) \Eu^T = (P \otimes Q^T) \Eu^T$ (and its transpose) \citep{Tan2024b}. Next, $S_{11} - S_{12} S_{22}^{-1} S_{21} $ is given by 
\[
\begin{aligned}
& \El  (U \otimes C^{-T}) \{N_\lambda  - (\K + I_d\otimes D_\kappa^{-2}) (I_d \otimes U^{-1})\Eu^T  \Eu  (I_d \otimes U) \Eu^T \\
& \times\diag\{ \vech_u(K)\} \Eu (I_d \otimes U^T)\Eu^T  \Eu (I_d \otimes U^{-T}) (\K + I_d\otimes D_\kappa^{-2}) \}(U^T \otimes C^{-1}) \El^T \\
& = \El  (U \otimes C^{-T}) B (U^T \otimes C^{-1}) \El^T,
\end{aligned}
\]
where $B = N_\lambda  - (\K + I_d\otimes D_\kappa^{-2})\Eu^T \diag\{ \vech_u(K)\} \Eu (\K + I_d\otimes D_\kappa^{-2}) $. It can be verified that $\diag\{ \vech_u(K)\} \Eu  (I_d\otimes D_\kappa^{-2}) = \Eu $. In addition, we have the properties 
\[
\El^T \El + \Eu^T \Eu = I_{d^2}, \quad 
\K  - \Eu^T \Eu \K - \K\Eu^T \Eu = \diag\{ \vec(I_d)\}.
\]
from \cite{Tan2024b}. Thus we can simplify
\[
\begin{aligned}
B &=  \K + I_d\otimes D_\kappa^{-2} - \Ed^T \diag(\tfrac{1-\kappa^2}{2-\kappa^2}) \Ed^T  - \K \Eu^T \diag\{ \vech_u(K) \Eu \K - \Eu^T \Eu \K - \K\Eu^T \Eu -\Eu^T \Eu (I_d\otimes D_\kappa^{-2})  \\
&= \diag[ \vec \{I_d - \diag ( \tfrac{1-\kappa^2}{2-\kappa^2} ) \} ]  + \El^T \El (I_d\otimes D_\kappa^{-2})  - \K \diag\{ \vec( K_u )\} \K \\
&= \diag[ \vec \{ \diag ( \tfrac{1}{2-\kappa^2} ) + (\kappa^{-2} \bfone^T)_\ell  -   (K_u)^T \} ] \\
&= \El^T \diag(a) \El,
\end{aligned}
\]
where $a =  \vech \{ \diag ( \tfrac{1}{2-\kappa^2} ) + (\kappa^{-2} \bfone^T)_\ell  -   (K_u)^T \}$. We claim that
\[
(S_{11} - S_{12} S_{22}^{-1} S_{21} )^{-1} = \El (I_d \otimes L)\El^T \El (U^{-T} \otimes U) \El^T \diag(1/a) \El(U^{-1} \otimes U^T) \El^T \El (I_d \otimes L^T) \El^T.
\]
since
\[
\begin{aligned}
&\{ \El  (U \otimes C^{-T})  \El^T \diag(a) \El (U^T \otimes C^{-1}) \El^T \} \{  \El (I_d \otimes L)\El^T \El (U^{-T} \otimes U) \El^T \diag(1/a) \El \\
& \quad \times (U^{-1} \otimes U^T) \El^T \El (I_d \otimes L^T) \El^T \} \\
&= \El  (U \otimes C^{-T})  \El^T \diag(a) \El(U^T \otimes U^{-1}) (I_d \otimes L^{-1} )(I_d \otimes L)\El^T \El (U^{-T} \otimes U)  \\
& \quad \times \El^T \diag(1/a) \El (U^{-1} \otimes U^T) \El^T \El (I_d \otimes L^T) \El^T \\
&= \El  (U \otimes C^{-T})  \El^T \diag(a) \El (U^T \otimes U^{-1}) (U^{-T} \otimes U) \El^T \diag(a) \El  (U^{-1} \otimes U^T) \El^T \El (I_d \otimes L^T) \El^T \\
&= \El  (U \otimes C^{-T})  \El^T \diag(a) \El \El^T \diag(a) \El  (U^{-1} \otimes U^T) \El^T \El (I_d \otimes L^T) \El^T \\
&= \El  (I \otimes L^{-T}) (U \otimes U^{-T}) \El^T \El  (U^{-1} \otimes U^T) \El^T \El (I_d \otimes L^T) \El^T \\
&= \El  (I \otimes L^{-T}) (U \otimes U^{-T}) (U^{-1} \otimes U^T) \El^T \El (I_d \otimes L^T) \El^T \\
&= \El  (I \otimes L^{-T}) \El^T \El (I_d \otimes L^T) \El^T \\
&= \El \El^T = I_{d(d+1)/2}.
\end{aligned}
\]
From \eqref{LU nat grad supp},
\[
\begin{bmatrix}
\widetilde{\nabla}_{\vech(L)} \mL \\ \widetilde{\nabla}_{\vech_u(U)} \mL 
\end{bmatrix}
= (\mathcal{I}_{33} - \mathcal{I}_{32} \mathcal{I}_{22}^{-1} \mathcal{I}_{23})^{-1} 
\begin{bmatrix} \nabla_{\vech(L)} \mL  + \vech \{C^{-T} \diag(\tfrac{\lambda}{2-\kappa^2} \odot \nabla_\lambda \mL) U^T\}  \\  \nabla_{\vech_u(U)} \mL \end{bmatrix}.
\]

As 
\[
\begin{aligned}
S_{12} S_{22}^{-1}  \nabla_{\vech_u(U)} \mL&=  \El  (U \otimes C^{-T}) \{ \K \Eu^T \diag\{ \vech_u (K) \} + \Eu^T \}  \Eu(I_d \otimes U^T) \Eu^T  \nabla_{\vech_u(U)} \mL \\
&=  \El  (U \otimes C^{-T}) \{ \K \Eu^T \diag\{ \vech_u (K) \} + \Eu^T \} \vech_u (F) \\
&=  \El  (U \otimes C^{-T}) \{ \K \vec( K \odot F) + \vec (F) \}  \\
&=  \El  (U \otimes C^{-T}) \vec \{( K \odot  F)^T+  F \}  \\
&=  \vech[C^{-T} \{( K \odot F)^T+  F \}  U^T] ,
\end{aligned}
\]
\[
\begin{aligned}
\widetilde{\nabla}_{\vech(L)} \mL &= (S_{11} - S_{12} S_{22}^{-1} S_{21})^{-1} [ \nabla_{\vech(L)} \mL + \vech\{ C^{-T} \diag(\tfrac{\lambda}{2-\kappa^2} \odot \nabla_\lambda \mL)  U^T\} - S_{12} S_{22}^{-1}  \nabla_{\vech_u(U)} \mL ] \\
&= \El (I_d \otimes L)\El^T \El (U^{-T} \otimes U) \El^T \diag(1/a) \El(U^{-1} \otimes U^T) \El^T \El (I_d \otimes L^T) \El^T \\
&\quad \times  ( \nabla_{\vech(L)} \mL + \vech[C^{-T} \{\diag (\tfrac{\lambda}{2-\kappa^2} \odot  \nabla_\lambda \mL)  - ( K \odot  F)^T -  F \} U^T ] )\\
&= \El (I_d \otimes L)\El^T \El (U^{-T} \otimes U) \El^T \diag(1/a) \El(U^{-1} \otimes U^T) \El^T \\
&\quad \times \vech[G + U^{-T} \{\diag (\tfrac{\lambda}{2-\kappa^2} \odot  \nabla_\lambda \mL)  - ( K \odot  F)^T -  F \} U^T ] \\
\end{aligned}
\]

\[
\begin{aligned}
&= \El (I_d \otimes L)\El^T \El (U^{-T} \otimes U) \El^T \diag(1/a)  \\
&\quad \times  \vech \{ U^T (G - U^{-T} F U^T)_\ell U^{-T} + \diag (\tfrac{\lambda}{2-\kappa^2} \odot  \nabla_\lambda \mL)  - ( K \odot  F)^T \}  \\
&= \El (I_d \otimes L)\El^T \El (U^{-T} \otimes U)  \\
&\quad \times  \vec [ \{ U^T (G - U^{-T} F U^T)_\ell U^{-T} + \diag (\tfrac{\lambda}{2-\kappa^2} \odot  \nabla_\lambda \mL)  - ( K \odot F)^T \} \odot A_2 ] \\
&=  \El (I_d \otimes L)\El^T \vech( U H U^{-1}) =  \vech\{L \mG_\ell\}.
\end{aligned}
\]
Finally, $\widetilde{\nabla}_{\vech_u(U)} \mL = S_{22}^{-1} \vech(F) - S_{22}^{-1} S_{21} \vech (L \mG_\ell)$ where
\[
\begin{aligned}
S_{22}^{-1} \nabla_{\vech_u(U)} \mL&= \Eu  (I_d \otimes U) \Eu^T \diag\{ \vech_u(K)\} \Eu (I_d \otimes U^T)\Eu^T \nabla_{\vech_u(U)} \mL \\
&= \Eu  (I_d \otimes U) \Eu^T \diag\{ \vech_u(K)\} \vech_u( F) \\
&= \Eu  (I_d \otimes U) \Eu^T \vech_u (K \odot  F) = \vech_u \{ U (K \odot F) \},
\end{aligned}
\]
and 
\[
\begin{aligned}
S_{22}^{-1} S_{21} \vech (L \mG_\ell) &= \Eu (I_d \otimes U) \Eu^T  \{ \Eu + \diag(\vech_u(K)) \E_u K \} (U^T \otimes C^{-1} ) \El^T  \vech\{L \mG_\ell\} \\
&= \Eu (I_d \otimes U) \Eu^T  \{ \Eu + \diag(\vech_u(K)) \E_u K \} \vec(U^{-1} \mG_\ell U) \\
&= \Eu (I_d \otimes U) \Eu^T [\vech^u(U^{-1} \mG_\ell U)  + \diag\{\vech_u(K)\} \vech^u\{ (U^{-1} \mG_\ell U)^T\}  ] \\
&= \Eu (I_d \otimes U) \Eu^T \vech^u \{ U^{-1} \mG_\ell U + K_u \odot (U^{-1} \mG_\ell U)^T \} \\
&= \vech_u [U \{ (U^{-1} \mG_\ell U)_u + K_u \odot (U^{-1} \mG_\ell U)^T \} ].  \\
\end{aligned}
\]
Hence
\begin{equation} \label{U update}
\widetilde{\nabla}_{\vech_u(U)} \mL = \vech_u [ U \{ K_u \odot ( F - (U^{-1} \mG_\ell U)^T ) -  (U^{-1} \mG_\ell U)_u \} ].
\end{equation}
From \eqref{LU nat grad supp},
\begin{equation} \label{lambda update}
\begin{aligned}
\widetilde{\nabla}_\lambda \mL 
&=  \tfrac{1}{(1-b^2)(2\kappa^2 - \kappa^4)} \odot  \nabla_\lambda \mL + \tfrac{1}{(2\kappa^2 - \kappa^4)} \odot \diag(\alpha \kappa) \Ed (U^T \otimes C^{-1}) \El^T   \vech(L \mG_\ell)  \\
&=  \tfrac{1}{(1-b^2)(2\kappa^2 - \kappa^4)} \odot  \nabla_\lambda \mL + \tfrac{\lambda}{2-\kappa^2} \odot \Ed \vec(U^{-1} \mG_\ell U)  \\
&= \tfrac{1}{(1-b^2)(2\kappa^2 - \kappa^4)} \odot  \nabla_\lambda \mL + \tfrac{\lambda}{2-\kappa^2} \odot \diag(U^{-1} \mG_\ell U).
\end{aligned}
\end{equation}
To simplify the above results, we apply results from \cite{Tan2024b}, which states that $\vech(U^{-1}\mG_\ell U) = \vech(H)$ and $\vech^u\{ U (U^{-1} \mG_\ell U)_u \} = - \vech^u(\mG_u U)$.  Thus we can replace $\diag(U^{-1} \mG_\ell U)$ in \eqref{lambda update} by $\diag(H)$, and $(U^{-1} \mG_\ell U)^T$ in \eqref{U update} by $H^T$ because only the upper triangular elements of $(U^{-1} \mG_\ell U)^T$ (equivalent to lower triangular elements of $U^{-1} \mG_\ell U$) will be extracted.

\section{Flow methods}
Flow methods provide rich and flexible variational approximations by transforming a simple density using a series of bijective and differentiable functions. Specifically, the density $q_k(\theta_k)$ for the $k$th flow is obtained by sequentially transforming an initial random variable $\theta_0$ with density $q_0$ such that $\theta_k = f_k \circ \cdots \circ f_1(\theta_0)$, and
\[
\begin{aligned}
\log q_k(\theta_k) &= \log q_0(\theta_0) - \sum_{i=1}^k \log \left|  \frac{\partial f_i}{\partial \theta_{i-1}} \right|.
\end{aligned}
\]

In the {\em planar flow} \citep{Rezende2015}, each transformation is of the form,
\[
f_i(\theta) = \theta + u_i h (w_i^\top \theta + b_i),
\]
where $w_i, u_i \in \mathbb{R}^d$ and $b_i \in \mathbb{R}$ are free parameters of the flow. This approach is so named as it uses the nonlinear differentiable function $h:\mathbb{R} \rightarrow \mathbb{R}$ to expand or contract a density in a direction that is orthogonal to the hyperplane defined by $w_i^\top \theta + b_i = 0$. The determinant of the Jacobian can be computed efficiently as 
\[
\left| \frac{\partial f}{\partial \theta} \right| = \left| I_d + h'(w_i^\top \theta + b_i) u_i w_i^\top \right| = \left| 1 + h'(w_i^\top \theta + b_i) w_i^\top u_i \right|
\]
using Sylvester's determinant theorem. However, invertibility of the flow depends on the choice of $h(\cdot)$. If $h(\cdot)$ is $\tanh(\cdot)$, then the inverse flow exists provided $w_i^\top u_i \geq 1$.

Real NVP flow \citep{Dinh2017} is a powerful class of flows that provides more expressive models for high-dimensional data, while maintaining exact likelihood computation, inference, sampling and inversion. Each transformation updates only a subset of the elements in $\theta$, which is defined via a binary mask $b_i$ of length $d$, such that
\[
f_i(\theta) = b_i\odot\theta+(1-b_i)\odot [\theta\odot \exp\{s_i (b\odot\theta)\}+ t_i(b\odot\theta)],
\]
where $s_i, t_i: \mathbb{R}^d \rightarrow \mathbb{R}^d$ are scale and translation functions respectively, and $\odot$ denotes elementwise multiplication. In this formulation, the Jacobian determinant can be computed efficiently because the determinant of a lower triangular matrix is simply the product of its diagonal entries. Moreover, $s_i$ and $t_i$ are allowed to be highly complex transformations that are hard to invert. By choosing them to be deep neural networks for instance, real NVP can provide expressive variational approximations efficiently for high-dimensional data. 

In {\tt normflows} \citep{Stimper2023}, normalizing flows are trained by minimizing a loss function (taken as the reverse KL divergence) using stochastic gradient descent and  the stepsize is computed using Adam. This is equivalent to maximizing the evidence lower bound in variational inference. From \cite{Papamakarios2021}, if the posterior density $p(\theta_k|y)$ is approximated by $q_k(\theta_k)$, then the loss function (reverse KL divergence) is 
\begin{multline*}
\int q_k(\theta_k)\log \frac{q_k(\theta_k)}{p(\theta_k|y)} d \theta_k = \log p(y) + \E_{q_k} [ \log q_k(\theta_k) - \log p(\theta_k, y) ] \\
= \log p(y) + \E_{q_o} \left[\log q_0(\theta_0) - \sum_{i=1}^k \log \left|\frac{\partial f_i}{\partial \theta_{i-1}} \right| - \log p(f_k \circ \cdots \circ f_1(\theta_0), y) \right].
\end{multline*}
The first term $\log p(y)$ is just a constant with respect to the flow parameters, and the gradient of the loss with respect to flow parameters can be estimated by automatic differentiation and Monte Carlo by generating samples from $q_0(\theta_0)$.

\section{Applications}
\subsection{Logistic regression}
The log joint density and its gradient are  
\[
\begin{aligned}
\log p(y, \theta) &= -\tfrac{d}{2} \log (2\pi \sigma_o^2) - \theta^T\theta/(2\sigma_0^2) + y^T X \theta - \sum_{i=1}^n n_i \log\{ 1 + \exp(x_i^T \theta)\} + \sum_{i=1}^n {n_i \choose y_i}\\
\nabla_\theta \log p(y, \theta) &= X^T y - \sum_{i=1}^n n_i p_i x_i -\theta/\sigma_0^2
\end{aligned}
\]
where $y=(y_1, \dots, y_n)^T$, $X=(x_1, \dots, x_n)^T$, $p = (p_1, \dots, p_n)^T$ and $p_i = \exp(x_i^T \theta)/\{1 + \exp(x_i^T \theta)\}$ for $i=1, \dots, n$. For binary responses, $n_i = 1$.

\subsection{Zero-inflated negative binomial model}
The joint log-likelihood is 
\begin{multline*}
\log p(y, \theta) = \sum_{i=1}^n [  \mathbbm{1}_{\{y_i=0\}}  \log\{ \e^{z_i^T \gamma} +  (\alpha \e^{x_i^T \beta} +1)^{-1/\alpha} \} - \log ( 1 + \e^{z_i^T \gamma})  + \mathbbm{1}_{\{y_i>0\}} \{ y_i x_i^T \beta   \\
- \tfrac{1}{\alpha} \log\alpha + \log \Gamma (y_i +  \tfrac{1}{\alpha}) - \log \Gamma( \tfrac{1}{\alpha}) - \log (y_i!) - (y_i +  \tfrac{1}{\alpha} ) \log (\e^{x_i^T \beta} +  \tfrac{1}{\alpha} )\}  ] - \tfrac{d}{2} \log (2\pi\sigma_0^2)  - \tfrac{\theta^T \theta}{2\sigma_0^2}.
\end{multline*}
The gradients are 
\[
\begin{aligned}
\nabla_\beta \log p(y, \theta) &= \sum_{i=1}^n \bigg[ \mathbbm{1}_{\{y_i>0\}} \bigg\{ y_i  -\frac{ (\alpha y_i +  1 )  \e^{x_i^T \beta}}{ \alpha \e^{x_i^T \beta} + 1 } \bigg\}  - \mathbbm{1}_{\{y_i=0\}} \frac{ \e^{x_i^T \beta} (\alpha \e^{x_i^T \beta} +1)^{-1/\alpha - 1} }{ \e^{z_i^T \gamma} +  (\alpha \e^{x_i^T \beta} +1)^{-1/\alpha}} \bigg] x_i - \tfrac{\beta}{\sigma_0^2}, \\
\nabla_\gamma\log p(y, \theta) &= \sum_{i=1}^n \bigg[   \frac{ \mathbbm{1}_{\{y_i=0\}} \e^{z_i^T \gamma} }{ \e^{z_i^T \gamma} +  (\alpha \e^{x_i^T \beta} +1)^{-1/\alpha} } - \frac{\e^{z_i^T \gamma}}{ 1 + \e^{z_i^T \gamma}} \bigg] z_i - \tfrac{\gamma}{\sigma_0^2},
\end{aligned}
\]
\begin{multline*}
\nabla_{\log \alpha} \log p(y, \theta) = (\nabla_{\log \alpha} \alpha) \nabla_\alpha \log p(y, \theta) \\
= \sum_{i=1}^n \bigg[  \frac{\mathbbm{1}_{\{y_i=0\}} \Big\{  \tfrac{1}{\alpha} \log  (\alpha \e^{x_i^T \beta} +1) - \frac{\e^{x_i^T \beta}}{  (\alpha \e^{x_i^T \beta} +1) } \Big\} }{ \e^{z_i^T \gamma}  (\alpha \e^{x_i^T \beta} +1)^{1/\alpha}  + 1}  \\
+ \tfrac{1}{\alpha} \mathbbm{1}_{\{y_i>0\}} \bigg\{  \psi( \tfrac{1}{\alpha}) -  \psi (y_i +  \tfrac{1}{\alpha}) - 1 + \log (\alpha \e^{x_i^T \beta}  + 1 ) + \frac{\alpha y_i + 1}{ \alpha \e^{x_i^T \beta} +  1 }  \bigg\}  \bigg] - \tfrac{\log \alpha}{\sigma_0^2}.
\end{multline*}

\subsection{Survival model}
Let the observed data $y = (t^T, d^T)$, where $t=(t_1, \dots, t_n)^T$ and $d = (d_1, \dots, d_n)^T$. The joint log-likelihood is 
\[
\begin{aligned}
\log p(y, \theta) &= \sum_{i=1}^n \{ d_i \log h(t_i) + \log S(t_i) \}- \tfrac{d}{2} \log (2\pi\sigma_0^2) -\tfrac{\theta^T \theta}{2\sigma_0^2} \\
&= \sum_{i=1}^n [ d_i \{ \log \rho_i  + (\rho_i-1) \log t_i + x_i^T \beta \} - \exp(x_i^T \beta) t_i^{\rho_i}] - \tfrac{d}{2} \log (2\pi\sigma_0^2) -\tfrac{\theta^T \theta}{2\sigma_0^2}  \\
&= \sum_{i=1}^n [ d_i \{ z_i^T \gamma  + \{ \exp(z_i^T \gamma) - 1\} \log t_i + x_i^T \beta \} - \exp(x_i^T \beta) t_i^{\exp(z_i^T \gamma)}]- \tfrac{d}{2} \log (2\pi\sigma_0^2) -\tfrac{\theta^T \theta}{2\sigma_0^2}.
\end{aligned}
\]
The gradients are 
\[
\begin{aligned}
\nabla_\beta \log p(y, \theta) &= \sum_{i=1}^n \{d_i - \exp(x_i^T \beta) t_i^{\exp(z_i^T \gamma)} \} x_i - \theta/\sigma_0^2. \\
\nabla_\gamma\log p(y, \theta) &= \sum_{i=1}^n [  d_i + \exp(z_i^T \gamma) \log t_i \{ d_i - \exp(x_i^T \beta) t_i^{\exp(z_i^T \gamma)} \}] z_i - \theta/\sigma_0^2.
\end{aligned}
\]

\subsection{Generalized linear mixed model}
The joint distribution can be written as 
\begin{equation*}
h(\theta) = p(y, \theta) = p(\beta) p(\zeta) \prod_{i=1}^n  \left\{ p(b_i|\zeta) \prod_{j=1}^{n_i} p(y_{ij}|\beta, b_i)  \right\}.
\end{equation*}
Here $\beta$ is the $p\times 1$ fixed effect and $b_i$ is the $r\times 1$ random effect of the $i$th subject. The precision matrix of random effects is decomposed as $G=WW^{\top}$ with $W$ being the lower triangular Cholesky factor. Define $W^*$ such that $W_{ii}^* = \log(W_{ii})$ for diagonal elements and $W_{ij}^* = W_{ij}$ for off-diagonal elements. Let $\zeta = \text{vech}(W^*)$ and assume $\zeta \sim \N(0, \sigma_\zeta^2 I_{r(r+1)/2})$ where $\sigma_\zeta=10$. We assume a normal prior $\beta \sim \N(0, \sigma_\beta^2 I)$, where $\sigma_\beta=10$. Thus $\theta_G = (\beta^\top, \zeta^\top)^\top$. Focusing on GLMMs with canonical links, the log joint density is
\begin{equation*}
\begin{aligned}
\log h(\theta) & = \sum_{i=1}^n \sum_{j=1}^{n_i} \log p(y_{ij}|\beta, b_i) + \sum_{i=1}^n \log p(b_i|\zeta) + \log p(\beta) + \log p(\zeta)  \\
& = \sum_{i,j} \{y_{ij} \eta_{ij} - A(\eta_{ij})\}  + n\log|W| -\frac{1}{2} \sum_{i=1}^n b_i^T WW^{\top}b_i -\frac{\beta^T\beta}{2\sigma_\beta^2} -\frac{\zeta^T\zeta}{2\sigma_\zeta^2} +C, 
\end{aligned}
\end{equation*}
where $\eta_{ij} = X_{ij}^T \beta + Z_{ij}^T b_i$, $C$ is a constant independent of $\theta$, and $A(\cdot)$ is the log-partition function. For Bernoulli GLMMs, $A(\eta_{ij}) = \log\{ 1 + \exp(\eta_{ij})\}$ and for Poisson GLMMs, $A(\eta_{ij}) = \exp(\eta_{ij})$. 
Let $D^W$ be the diagonal matrix where the diagonal 
is given by $\vech(J^W)$, and $J^W$ is an $r\times r$ lower triangular matrix with the $i$th diagonal entry being $W_{ii}$ and all off-diagonal entries being 1. The gradient of $\log h(\theta)$ is given by $[\nabla_{b_1} \log h(\theta), \dots, \nabla_{b_n} \log h(\theta), \nabla_{\beta} \log h(\theta), \nabla_{\zeta} \log h(\theta)] $, where
\begin{equation*} 
\begin{aligned}
\nabla_{b_i} \log h(\theta) &= \sum_{j=1}^{n_i} \{y_{ij} - A'(\eta_{ij})  \}Z_{ij} -G b_i, \text{ for } i=1,\dots, n,\\
\nabla_{\beta} \log h(\theta) &= \sum_{i=1}^n \sum_{j=1}^{n_i} \{y_{ij} - A'(\eta_{ij}) \}X_{ij} -\frac{\beta}{\sigma_{\beta}^2},\\
\nabla_{\zeta} \log h(\theta) &=  -D^W\sum_{i=1}^{n}\vech(b_ib_i^TW) +n\vech(I_r)   -\frac{\zeta}{\sigma_\zeta^2}.
\end{aligned}
\end{equation*}

\spacingset{1.3} 

{
\footnotesize
\bibliographystyle{chicago}
\bibliography{ref}
}

\end{document}